\documentclass[a4paper,11pt]{article}
\usepackage[margin=2cm]{geometry}

\usepackage[utf8]{inputenc}
\usepackage{amsmath}
\usepackage{amssymb}
\usepackage{amsthm}
\usepackage{graphicx}
\usepackage{float}
\usepackage{xcolor}
\usepackage{cite}
\usepackage{microtype}
\usepackage[shortlabels]{enumitem}
\usepackage{hyperref}

\newtheorem{lemma}{Lemma}
\newtheorem{corollary}{Corollary}

\usepackage{thmtools, thm-restate}
\declaretheorem{theorem}

\newcommand{\softO}{\widetilde{O}}
\newcommand{\floor}[1]{\left\lfloor #1 \right\rfloor}

\newcommand{\dle}[4]{d_{#1}^{\leqslant{#2}}({#3},{#4})}

\let\epsilon\varepsilon

\usepackage[ruled,vlined,linesnumbered]{algorithm2e}
\title{Sparse Temporal Spanners with Low Stretch}

\author{
  Davide~Bilò\\
  Department of Information Engineering, Computer Science and Mathematics,\\ University of L'Aquila, Italy\\
  \texttt{davide.bilo@univaq.it}
  \and
  Gianlorenzo~D'Angelo\\
  Gran Sasso Science Institute, L'Aquila, Italy\\
  \texttt{gianlorenzo.dangelo@gssi.it}
  \and
  Luciano~Gualà\\
  Department of Enterprise Engineering, University of Rome ``Tor Vergata'', Italy\\
  \texttt{guala@mat.uniroma2.it}
  \and
  Stefano~Leucci\\
  Department of Information Engineering, Computer Science and Mathematics,\\ University of L'Aquila, Italy\\
  \texttt{stefano.leucci@univaq.it}
  \and
  Mirko~Rossi\\
  Gran Sasso Science Institute, L'Aquila, Italy\\
  \texttt{mirko.rossi@gssi.it}
}
\date{}

\begin{document}
\maketitle

\begin{abstract}
A \emph{temporal graph} is an undirected graph $G=(V,E)$ along with a function $\lambda : E \to \mathbb{N}^+$ that assigns a time-label to each edge in $E$.
A path in $G$ such that the traversed time-labels are non-decreasing is called \emph{temporal path}. Accordingly, the distance from $u$ to $v$ is the minimum length (i.e., the number of edges) of a temporal path from $u$ to $v$.
A temporal $\alpha$-spanner of $G$ is a (temporal) subgraph $H$ that preserves the distances between any pair of vertices in $V$, up to a \emph{multiplicative} stretch factor of $\alpha$. The \emph{size} of $H$ is measured as the number of its edges.

In this work, we study the size-stretch trade-offs of temporal spanners.
In particular we show that temporal cliques always admit a temporal $(2k-1)-$spanner with $\softO(kn^{1+\frac{1}{k}})$ edges, where $k>1$ is an integer parameter of choice.
Choosing $k=\lfloor \log n \rfloor$, we obtain a temporal $O(\log n)$-spanner with $\softO(n)$ edges that has almost the same size (up to logarithmic factors) as the temporal spanner given in [Casteigts et al., JCSS 2021] which only preserves temporal connectivity.

We then turn our attention to \emph{general} temporal graphs.  Since $\Omega(n^2)$ edges might be needed by any connectivity-preserving temporal subgraph [Axiotis et al., ICALP'16], we focus on approximating distances from a \emph{single source}. We show that $\softO(n/\log(1+\varepsilon))$ edges suffice to obtain a stretch of $(1+\varepsilon)$, for any small $\varepsilon > 0$. This result is essentially tight in the following sense: there are temporal graphs $G$ for which any temporal subgraph preserving exact distances from a single-source must use $\Omega(n^2)$ edges.
Interestingly enough, our analysis can be extended to the case of \emph{additive} stretch for which we prove an upper bound of $\softO(n^2 / \beta)$ on the size of any temporal $\beta$-additive spanner, which we show to be tight up to polylogarithmic factors.

Finally, we investigate how the \emph{lifetime} of $G$, i.e., the number of its distinct time-labels, affects the trade-off between the size and the stretch of a temporal spanner.
\end{abstract}

\section{Introduction}

A \emph{temporal graph} is a graph $G=(V,E)$ in which each edge can be used only in certain time instants. This recurrent idea of time-evolving graphs has been formalized in multiple ways, and a simple widely-adopted model is the one of Kempe, Kleinberg, and Kumar~\cite{KempeKK02}, in which each edge $e \in E$ has an assigned time-label $\lambda(e)$ representing the instant in which $e$ can be used. 
A path from a vertex to another in $G$ is said to be a \emph{temporal path} if the time-labels of the traversed edges are non-decreasing.
Accordingly, a graph is  \emph{temporally connected} if there exists a temporal path from $u$ to $v$, for every two vertices $u,v \in V$.

Notice that, unlike paths in \emph{static} graphs, the existence of temporal paths is neither symmetric nor transitive.\footnote{Indeed, a temporal path from $u$ to $v$ is not necessarily a temporal path from $v$ to $u$, even when $G$ is undirected. Moreover, the existence of a temporal path from $u$ to $v$, and of a  temporal path from $v$ to $w$, does not imply the existence of a temporal path from $u$ to $w$.}
For this reason, temporal graphs exhibit a different combinatorial structure compared to static graphs, and even problems that admit easy solutions on static graphs become more challenging in their temporal counterpart.
Indeed, one of the main problems introduced in the seminal paper of Kempe, Kleinberg, and Kumar~\cite{KempeKK02} is that of finding a sparse temporally connected subgraph $H$ of an input temporal graph $G$. Such a subgraph $H$ is sometimes referred to as a \emph{temporal spanner} of $G$.
While any spanning-tree is trivially a connectivity-preserving subgraph of a \emph{static} graph, not all temporal graphs $G$ admit a temporal spanner having $O(n)$ edges \cite{KempeKK02}. In particular, \cite{KempeKK02} exhibits a class of temporal graphs that contain $\Theta(n \log n)$ edges and cannot be further sparsified. Later, \cite{AxiotisF16} provided a stronger negative result showing that there are temporal graphs $G$ such that any temporal spanner of $G$ must use $\Theta(n^2)$ edges.
These strong lower bounds on general graphs motivated \cite{CasteigtsPS21} to focus on \emph{temporal cliques} instead. Here the situation improves significantly, as only $O(n \log n)$ edges are sufficient to guarantee temporal connectivity. This gives rise to the following natural question, which is exactly the focus of our paper: \emph{can one design a temporal spanner that also guarantees short temporal paths between any pair of vertices?}

To address this question, we measure the \emph{length} of a temporal path as the number of its edges,\footnote{Notice that, alternative definitions for the length of a temporal path are also natural, e.g., the \emph{arrival time}, \emph{departure time}, \emph{duration}, or \emph{travel time}. We discuss the corresponding distance measures in the conclusions.}
and we introduce
%
the notion of temporal $\alpha$-spanner of a temporal graph $G$, i.e., a subgraph $H$ of $G$ such that $d_H(u,v) \le \alpha \cdot d_G(u,v)$ for every pair of vertices $u,v \in V$, where $d_H(u,v)$ (resp. $d_G(u,v)$) denotes the length of a shortest temporal path from $u$ to $v$ in $H$ (resp. $G$).
Our main question then becomes that of understanding which trade-offs can be achieved between the \emph{size}, i.e., the number of edges, of $H$ and the value of its \emph{stretch-factor} $\alpha$.
This same question received considerable attention on static graphs and gave rise to a significant amount of work (see, e.g., \cite{AhmedBSHJKS20}), hence we deem investigating its temporal counterpart as a very interesting research direction.

To the best of our knowledge, the only temporal $\alpha$-spanner currently known
is actually the connectivity-preserving subgraph of \cite{CasteigtsPS21} having size $O(n \log n)$. However, a closer inspection of its construction shows that the resulting $\alpha$-spanner can have stretch $\alpha = \Theta(n)$. 
In particular, even the problem of achieving stretch  $o(n)$ using $o(n^2)$ edges remains open.

In this paper we investigate which size-stretch trade-offs can be attained by selecting subgraphs of temporal graphs, as detailed in the following.

\subsection{Our results}

\subparagraph*{Temporal cliques.}

Following~\cite{CasteigtsPS21}, we start by considering temporal cliques (see Section~\ref{sec:all_pair}). 
Our main result is the following: given a temporal clique $G$ and an integer $k \ge 2$, we can construct, in polynomial time, a temporal $(2k-1)$-spanner of $G$ having size $O(k n^{1+1/k} \log^{1-1/k} n)$.
Interestingly, the special case $k= \lfloor \log n \rfloor$ shows that 
$O(n \log^2 n)$ edges suffice to ensure that a temporal path of length $O(\log n)$ exists between any pair of vertices.  For this choice of $k$, the size of our spanner is  only a logarithmic factor away from the size the temporal spanner of \cite{CasteigtsPS21} that uses $O(n \log n)$ edges and only preserves connectivity.

We also show that there are temporal cliques for which any temporal spanner with stretch smaller than $3$ must have $\Omega(n^2)$ edges.

\subparagraph*{Single-source temporal spanners on general graphs.}

Next, in Section~\ref{sec:single_source}, we move our attention from temporal cliques to general temporal graphs. As already pointed out, there are temporal graphs that do not admit any connectivity-preserving subgraph with $o(n^2)$ edges \cite{AxiotisF16}. Hence, we consider the special case in which we have a single source $s$.
One can observe that any temporal graph $G$ admits a temporal subgraph containing $O(n)$ edges and preserving the connectivity from $s$ (see also \cite{KempeKK02}). However, to the best of our knowledge, no non-trivial result is known on the size of subgraphs preserving approximate distances from $s$. 

We formalize this problem by introducing the notion of  \emph{single-source temporal $\alpha$-spanner} of $G=(V,E)$ w.r.t.\ a source $s \in V$, which we define as a subgraph $H$ of $G$ such that $d_H(s,v) \le \alpha \cdot d_G(s,v)$ for every $v \in V$. 
Our main contribution for the single-source case is the following: given any temporal graph $G$, we can compute in polynomial time a single-source temporal $(1+\epsilon)$-spanner having size $O(\frac{n \log^4 n}{ \log(1+\epsilon)} )$, where $\varepsilon > 0$ is a parameter of choice.

Furthermore, we show that any single-source temporal $1$-spanner (i.e., a subgraph preserving \emph{exact} distances from $s$) must have $\Omega(n^2)$ edges in general. 
Our construction can be generalized to provide a lower bound of $\Omega(\frac{n^2}{\beta})$ on the size of any single-source temporal $\beta$-additive spanner, namely a subgraph $H$ that preserves single-source distances up to an \emph{additive} term of at most $\beta \ge 1$ (i.e., we require $d_H(s,v) \le d_G(s,v) + \beta$ for all $v \in V$).

Interestingly, the same techniques used to obtain our single-source temporal $(1+\varepsilon)$-spanner can be also applied to build a single-source temporal $\beta$-additive spanner of size $O(\frac{n^2 \log^4 n}{\beta} )$, which essentially matches our aforementioned lower bound.

\subparagraph*{The role of lifetime.}

An important parameter that measures how time-dependent is a temporal graph $G=(V,E)$ is its \emph{lifetime}, i.e., the number
$L$ of distinct time-labels associated with the edges of $G$.
Indeed, a temporal graph with lifetime $L=1$ is just a \emph{static} graph, while any temporal graph trivially satisfies $L=O(n^2)$. 
It is not surprising that the lifetime plays a crucial role in determining the number of edges required by temporal spanners.
For example, the lower bound of $\Omega(n^2)$ on the size of any connectivity-preserving temporal subgraph requires $L = \Omega(n)$ \cite{AxiotisF16}. 
In this paper, we also present a collection of results with the goal of shedding some light on the lifetime-size trade-off of temporal spanners. In particular, our results provide the following lifetime-dependant upper bounds on the size of temporal $\alpha$-spanners
\begin{itemize}
    \item As far as temporal cliques are concerned, we show how to build, in polynomial time, a temporal $3$-spanner with $O(2^L n \log n)$ edges. This implies that, when $L=O(1)$, we can achieve stretch $3$ with $\softO(n)$ edges.\footnote{The notation $\softO(f(n))$ is a synonym for $O(f(n) \cdot \operatorname{polylog} f(n))$.}
    \item If $L=2$, we can find (in polynomial time) a temporal $2$-spanner of a temporal clique having size $O(n \log n)$. 
    We deem this result interesting since, as soon as $L > 2$, our lower bound of $\Omega(n^2)$ on the size of any temporal $2$-spanner still applies.
    \item We show that, when $L$ is small, \emph{general} temporal graphs can be sparsified by exploiting known size-stretch trade-offs for spanners of static graphs. In particular, we show that if it is possible to compute, in polynomial time, an $\alpha$-spanner of a static graph having size $f(n)$, then one can also build a temporal $\alpha$-spanner of size $O(L f(n))$.
    This yields, e.g., a temporal $\lfloor \log n \rfloor$-spanner of size $o(n^2)$ on general temporal graphs with $L=o(n)$.

\end{itemize}

\subsection{Related work}
The definitions of temporal graphs and temporal paths given in the literature sometimes differ from the ones we adopt here. We now discuss how our results relate to some of the most common variants.
A first difference concerns the notion of temporal paths: some authors consider \emph{strict temporal paths} \cite{KempeKK02,CasteigtsPS21,AkridaGMS17}, i.e., temporal paths in which edge labels must be strictly increasing (rather than non-decreasing).
As observed by \cite{KempeKK02}, if we adopt strict temporal paths then there are dense graphs that cannot be sparsified, indeed no edge can be removed from a temporal clique in which all edges have the same time-label. 
As observed in~\cite{CasteigtsPS21}, one can get rid of these problematic instances by assuming that time-labels are \emph{locally distinct}, namely that all the time-labels of the edges incident to any single vertex are distinct.
In this case all temporal paths are also strict temporal paths and hence they focus on temporal paths as defined in our paper.
A second difference concerns whether edges are allowed to have multiple time-labels, as in \cite{MertziosMS19, AkridaGMS17}. In this case, each edge $e$ is associated to a non-empty set of time instants $\lambda(e) \subseteq \mathbb{N}^+$ in which $e$ is available.
We observe that any algorithm that sparsifies a temporal clique with single time-labels can be directly used on the case of multiple time-labels by selecting an arbitrary time-label for each edge (see also the discussion in~\cite{CasteigtsPS21}).
This is no longer true when we consider general temporal graphs, since removing edge labels might affect distances.
However, all our algorithms work also in the case of multiple labels and, since our lower bounds are given for single labels, they also apply to the case of multiple labels.

Another research line concerns random temporal graphs.
In particular, temporal cliques in which each edge has a single time-label chosen u.a.r.\ from the set $\{1, \dots, \alpha\}$, where $\alpha \ge 4$, admit temporal spanners with $O(n \log n)$ edges w.h.p.\ \cite{AkridaGMS17}.
In \cite{CasteigtsRRZ20}, the authors study connectivity properties of random temporal graphs defined as an Erd\H{o}s-R\'enyi graph $G_{n,p}$ in which each edge $e$ has time-label chosen as the rank of $e$ in a random permutation of the graph's edges.
They show that $p=\frac{\log n}{n}$, $p=\frac{2\log n}{n}$, $p=\frac{3\log n}{n}$, and $p=\frac{4\log n}{n}$ are sharp thresholds to guarantee that the resulting temporal graph $G$ satisfies the following respective conditions asymptotically almost surely: 
a fixed pair of vertices can reach each other via temporal paths in $G$, 
there is some vertex $s$ which can reach all other vertices in $G$ via temporal paths, $G$ is temporally connected, $G$ and admits a temporal spanner with $2n-4$ edges (which is tight when time-labels are locally distinct).

Besides temporal graphs, other models to represent graphs or paths that evolve over time have been considered in the literature,  we refer the interested reader to~\cite{Holme18} for a survey. 

Finally, as we already mentioned, there is a large body of literature concerning spanners on static graphs, see \cite{AhmedBSHJKS20} for a survey on the topic.
A reader that is already familiar with the area might notice that our upper bound of $\softO(n^{1+\frac{1}{k}})$ on the size of a temporal $(2k-1)$-spanner of a \emph{temporal clique}, happens to resemble the classical upper bound of $O(n^{1+\frac{1}{k}})$ on the size of a $(2k-1)$-spanner of a \emph{general static graph} \cite{AlthoferDDJS93}.
Nevertheless, the first result only applies to complete (temporal) graphs and is obtained using different technical tools.

\section{Model and preliminaries}
\label{sec:model}
Let $G=(V,E)$ be an undirected \emph{temporal} graph with $n$ vertices, and a labeling function $\lambda : E \rightarrow \mathbb{N}^+$ that assigns a \emph{time-label} $\lambda(e)$ to each edge $e$.  If $G$ is complete we will say that it is a \emph{temporal clique}.
A temporal path $\pi$ from vertex $u$ to vertex $v$ is a path in $G$ from $u$ to $v$ such that the sequence $e_1, e_2, \dots, e_k$ of edges traversed by $\pi$ satisfies $\lambda(e_i) \le \lambda(e_{i+1})$ for all $i=1, \dots, k-1$.
We denote with $|\pi|$ the length of the $\pi$, i.e., the number of its edges. A shortest temporal path from vertex $u$ to vertex $v$ is a temporal path from $u$ to $v$ with minimum length. We denote with $d_G(u,v)$, the length of a shortest temporal path from $u$ to $v$ in $G$.
Given a generic graph $H$, we denote by $V(H)$ its vertex-set and by $E(H)$ its edge-set. 

For $\alpha \geq 1$ and $\beta \geq 0$, a temporal $(\alpha,\beta)$-spanner of $G$ is a (temporal) subgraph $H$ of $G$ such that $V(H)=V$ and $d_{H}(u,v) \leq \alpha \cdot d_{G}(u,v) + \beta$, for each $u,v \in V$.
We call a temporal $(\alpha,\beta)$-spanner: (i) temporal $\alpha$-spanner if $\beta = 0$, (ii) temporal $\beta$-additive spanner if $\alpha = 1$, (iii) temporal preserver if $\alpha = 1$ and $\beta = 0$.
We say that $H$ is a \emph{single-source} temporal $(\alpha,\beta)$-spanner w.r.t.\ a vertex $s \in V$, if $d_{H}(s,v) \leq \alpha \cdot d_{G}(s,v) + \beta $, for each $v \in V$.
The \emph{size} of a temporal spanner is the number of its edges.

We define the \emph{lifetime} $L$  of $G$ as the number of distinct time-labels of its edges. 
Furthermore, we assume w.l.o.g.\ that each time instant in $\{1, \dots, L\}$ is used by at least one time-label (since otherwise we can replace each time-label with its rank in the set $\{ \lambda(e) \mid e \in E \}$), so that $L = \max_{e \in E} \lambda(e)$. 

We will make use of the following well-known result (whose proof is provided for the sake of completeness). 

\begin{lemma}
\label{lemma:hitting_set}
Given a collection $\mathcal{S}$ of subsets of $\{1, \ldots , n\}$, where each subset has size at least $\ell$ and $|\mathcal{S}|$ is polynomially bounded in $n$, we can find in polynomial time a subset $R \subseteq \{1, \ldots , n\}$ of size $O((n/\ell) \log n)$ that hits all subsets in the collection, i.e., $R \cap S \neq \emptyset$ for all $S \in \mathcal{S}$. 
\end{lemma} 
\begin{proof}
Consider an iterative greedy algorithm that builds a sequence of partial hitting sets $R_i$ while keeping track of number $n_i$ of sets $S \in \mathcal{S}$ that are not hit by $R_i$.
Initially $R_0 = \emptyset$ and $n_0 = |\mathcal{S}|$.
In the generic $i$-th iteration, the algorithm finds an element $j \in \{1, \dots, n\}$ maximizing the number $c_i(j)$ of sets $S \in \mathcal{S}$ such that $j \in S$ and $S \cap R_{i-1} = \emptyset$ and adds it to the next partial hitting set, i.e., it sets $R_i = R_{i-1} \cup \{j\}$ and $n_i = n_{i-1} - c_i(j)$.  The algorithm stops as soon as $n_i =0$ and returns $R_i$.

In the rest of the proof we show that the number of iterations of the algorithm is at most $\frac{n}{\ell} \ln \mathcal{|S|} = O(\frac{n}{\ell} \log n)$, thus simultaneously bounding the running time of the algorithm (notice that each iteration can be performed in polynomial-time) and the size of the returned hitting set.  
At the beginning of the $i$-th iteration, there are $n_{i-1}$ sets in $\mathcal{S}$ that are not hit by $R_{i-1}$  and each occurrence of an element $j'$ in any such set contributes $1$ to $c_i(j')$. Since each set contains at least $\ell$ elements, we have $\sum_{j'=1}^n c_i(j') \ge n_{i-1}\ell$ implying that $c_i(j) = \max_{j'=1, \dots, n} c_i(j') \ge \frac{1}{n} \sum_{j'=1}^n c_i(j') \ge \frac{n_{i-1} \ell}{n}$ and hence $n_i = n_{i-1} -  c_i(j) \le n_{i-1} \cdot (1-\frac{\ell}{n}) \le n_0  \cdot (1-\frac{\ell}{n})^i = \mathcal{|S|} \cdot \big((1-\frac{\ell}{n})^n\big)^{i/n} \le |\mathcal{S}| \cdot e^{-\ell i / n}$.
As a consequence, there must be some
$i \le \frac{n}{\ell} \ln |\mathcal{S}|$ such that $n_i = 0$. Indeed, for $i > \frac{n}{\ell} \ln |\mathcal{S}|$, we have $|\mathcal{S}| \cdot e^{-\ell i / n} < 1$.
\end{proof}

\section{Spanners for temporal cliques}
\label{sec:all_pair}

In this section, we design an algorithm such that, given a temporal clique $G$, returns a temporal $(2k-1)$-spanner $H$ of $G$ with size $\softO(n^{1+\frac{1}{k}})$, for any integer $k>1$. We also provide a temporal clique $G$ for which any temporal $2$-spanner of $G$ has size $\Omega(n^2)$.

Before describing the algorithm for constructing temporal $(2k-1)$-spanners, we show as a warm up how to construct a temporal $3$-spanner and a temporal $5$-spanner of size $\softO(n^{1+\frac{1}{2}})$ and $\softO(n^{1+\frac{1}{3}})$, respectively.
\subsection{Our temporal \texorpdfstring{\boldmath $3$}{3}-spanner}
\begin{figure}
    \centering
    \includegraphics[scale=.9]{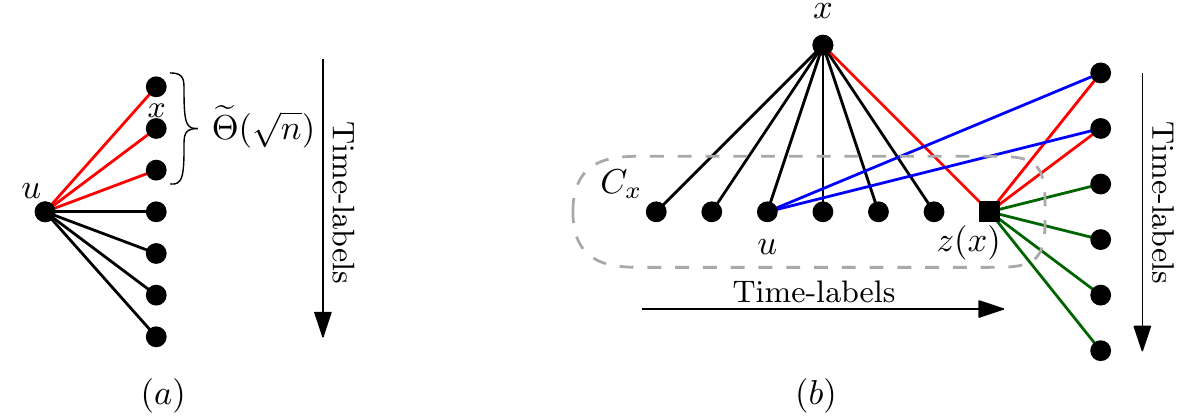}
    \caption{(a) A vertex $u$ and its neighbours in $G$, the edges are sorted top-down by increasing time-label. The red edges are those belonging to $E_u$ while the black edges belong to $E(G) \setminus E_u$. (b) An example of a cluster $C_x$ where $x \in R$. The edges incident to $x$ are sorted from left to right by increasing-time label. The black and red edges are added to $E(H)$ during the initialization phase, in particular the red edges are those belonging to $E_{z(x)}$. The blue edges are added w.r.t.\ $u$ to $E(H)$ during the first augmentation. The green edges belong to $E(G) \setminus E_{z(x)}$ and are added to $E(H)$ during the second augmentation.}
    \label{fig:clique_3spanner}
\end{figure}
\label{sec:3-spanner}
Given a temporal clique $G$, we construct a temporal $3$-spanner $H$ of $G$ via a clustering technique. 
For each $u \in V$, we select a set $E_u$ containing the $\Theta(\sqrt{n \log n})$ edges incident to $u$ having the smallest labels (ties are broken arbitrarily). We define $S_u = \{v \in V \mid (u,v) \in E_u\}$. 
Next, we find a hitting set $R \subseteq V$ of the collection $\{S_u\}_{u \in V}$. 
Thanks to Lemma~\ref{lemma:hitting_set}, we can deterministically compute a hitting set of size $|R|=O(\sqrt{n \log n})$.

We partition the vertices of $V$ into $|R|$ clusters. More precisely, we create a cluster $C_x \subseteq V$ for each vertex $x \in R$. Each vertex $u \in V$ belongs to exactly one arbitrarily chosen cluster $C_x$ that satisfies $x \in S_u$, i.e., $x$ hits $S_{u}$.
We call $x$ the center of cluster $C_x$.\footnote{Here and throughout the paper, the center of a cluster is not required to belong to the cluster itself.} Moreover, we choose the special vertex of cluster $C_x$ as a vertex $z(x)$ in $C_x$ that maximizes the label of the edge $(x,z(x))$.

Notice that, for every $x \in R$ and $u \in C_x$, $u$ can reach $z(x)$ via a temporal path of length at most $2$ in $G$ by using the edges $(u,x)$ and $(x,z(x))$ since, by definition of $z(x)$ and $S_u$, we have $\lambda(u,x) \leq \lambda(x,z(x))$.

We now build our temporal spanner $H$ of $G$. The set of edges $E(H)$ is constructed in three phases (See Figure~\ref{fig:clique_3spanner} for an example of the whole construction):
\begin{description}
\item[Initialization:] For each $u \in V$, we add the edges in $E_u$ to $E(H)$;
\item[First Augmentation:] For every $u \in V$, we add the edges in $E_{u,z(x)} = \{u\} \times S_{z(x)}$ to $E(H)$, where $x$ is the center of the cluster containing $u$;
\item[Second Augmentation:] For each $x \in R$, we add the edges in $E_{z(x),V} = \{z(x)\} \times V$ to $E(H)$.
\end{description}
It is easy to see that $H$ contains $O(n \sqrt{n \log n})$ edges. We now show that for any $u,v \in V$ there is a temporal path from $u$ to $v$ of length at most $3$ in $H$. Indeed, let $x \in R$ be the center of the cluster $C_x$ containing $u$. 
If $v=z(x)$ then, since $u \in C_x$, the initialization phase ensures that $(u, x) \in E(H)$ and $(x, z(x)) \in E(H)$, which form a temporal path as we already discussed above. 
We hence assume that $v\neq z(x)$.
If $(z(x),v) \in E_{z(x)}$ then the first augmentation phase added $(u,v) \in E_{u,z(x)}$ to $E(H)$, which is a temporal path of length one from $u$ to $v$. 
Otherwise $(z(x),v) \in E(G) \setminus E_{z(x)}$ and,  the second augmentation phase added edge $(z(x),v)$ to $E(H)$. Moreover, since $(z(x),v) \not\in E_{z(x)}$, $(z(x),v)$ is not among the $\Theta(\sqrt{n\log n})$ edges incident to $z(x)$ with lowest labels. As a consequence, since $(x,z(x)) \in E_{z(x)}$, we have $\lambda(x,z(x)) \leq \lambda(z(x),v)$. Hence, the edges $(u,x)$, $(x,z(x))$, and $(z(x),v)$ form a temporal path of length $3$ from $u$ to $v$ in $H$.

\subsection{Our temporal \texorpdfstring{\boldmath $5$}{5}-spanner}
\label{sec:5-spanner}
\begin{figure}
    \centering
    \includegraphics[width=\textwidth]{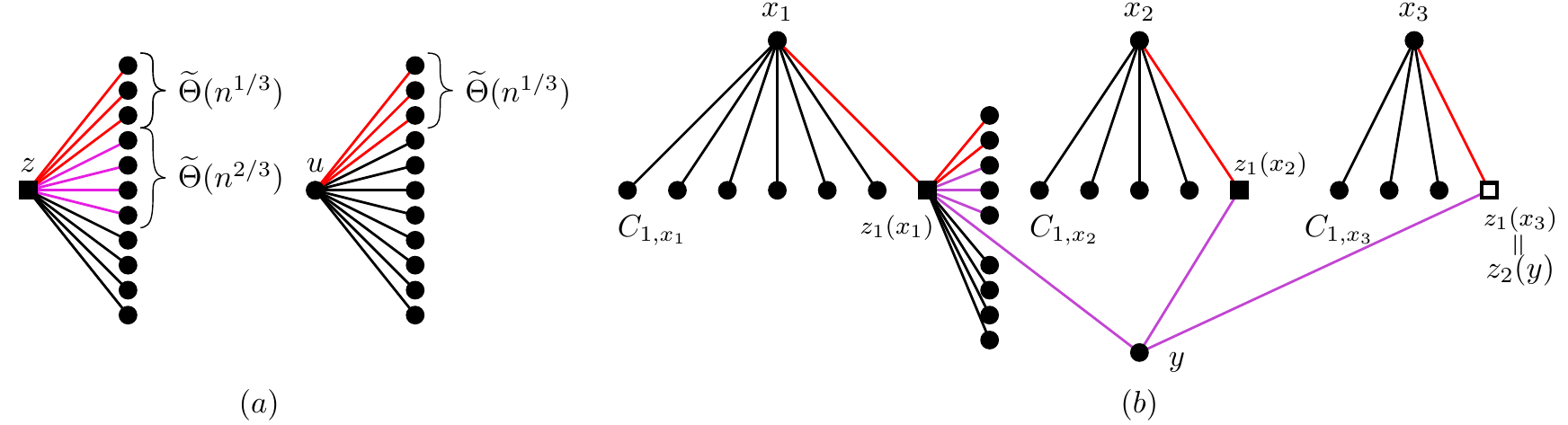}
    \caption{(a) Two vertices $u$ and $z$ of $G$, where $z \in Z_1$ and the red edges belong to $E_{1,u}$ and $E_{1,z}$, respectively. For vertex $z$ the purple edges belong to $E_{2,z}$. (b) A two level clustering. The level one consists in three cluster $C_{1,x_1}$, $C_{1,x_2}$, $C_{2,x_3}$. The level-two cluster $C_{2,y}$, with $y \in R_2$, contains vertices $z_1(x_1)$, $z_1(x_2)$ and $z_1(x_3)$, where $z_2(y) = z_1(x_3)$.}
    \label{fig:2levelClustering}
\end{figure}
 We show how to modify the construction of a temporal $3$-spanner given in previous section in order to obtain a temporal $5$-spanner of size $\softO(n^{4/3})$. The idea is to replace the single-level clustering of Section \ref{sec:3-spanner} with a two-level clustering, where the second-level clustering partitions the special vertices of the first level clustering and the number of selected clusters decreases as we move from the first level to the second one.

The level-one clustering is built similarly to the one used in our temporal $3$-spanner. For each vertex $u\in V$ we define sets $E_{1,u}$ and $S_{1,u}$ where $E_{1,u}$ consists of the $\Theta( n^{1/3} \log^{2/3} n )$ edges with the smallest label among those incident to $u$ (ties are broken arbitrarily) and $S_{1,u}= \{v \in V \mid (u,v) \in E_{1,u}\}$. We compute a hitting set $R_1$ of the collection $\{S_{1,u}\}_{u\in V}$, where $R_1$ has size $O(n^{2/3} \log^{1/3} n)$ thanks to Lemma~\ref{lemma:hitting_set}. We partition the vertices of $V$ into $|R_1|$ clusters $C_{1,x}$, for each $x\in R_1$, as before, and let $z_1(x)$ the vertex in $C_{1,x}$ that maximizes the label of the edge $(x,z_1(x))$.

The level-two clustering is built on top of the vertices $Z_1 = \{z_1(x) \mid x \in R_1\}$.
For each $u \in Z_1$, we define $E_{2,u}$ as a set of $\Theta( n^{2/3} \log^{1/3} n)$ edges with the smallest label among those that are incident to $u$ but do not belong to $E_{1,u}$. We also define a corresponding set $S_{2,u} = \{v \in V \mid (u,v) \in E_{2,u}\}$. 
We once again invoke Lemma~\ref{lemma:hitting_set} to compute a hitting set $R_2$ of size $O(n^{1/3} \log^{2/3} n)$
of the collection $\{S_{2,u}\}_{u \in Z_1}$. Based on $R_2$, we partition the special vertices in $Z_1$ by associating each $u \in Z_1$ to an arbitrary cluster $C_{2,x}$ centered in $x \in R_2$ such that $x \in S_{2,u}$.
Each cluster $C_{2,x}$ has an associated special vertex $z_2(x) \in C_{2,x}$ chosen among the ones that maximize the label of the edge $(x,z_2(x))$, see Figure~\ref{fig:2levelClustering}.

We are now ready to build our temporal $5$-spanner $H$.
As before, the set of edges $E(H)$ is constructed in three phases:
\begin{description}
\item[Initialization:] For each $u \in V$, we add the edges in $E_{1,u}$ to $E(H)$ and, for each $u \in Z_1$, we add the edges in $E_{2,u}$ to $E(H)$;
\item[First Augmentation:] For every $u \in V$, we add the edges in $E_{u,z_1(x)} = \{u\} \times S_{1,z_1(x)}$ to $E(H)$, where $x$ is the center of the cluster containing $u$. Moreover, for each $z \in Z_1$ we add the edges in $\{z\} \times (S_{1,z_2(x)} \cup S_{2,z_2(x)})$ to $E(H)$, where $x$ is the center of the level-two cluster $C_{2,x}$ containing $z$;
\item[Second Augmentation:] For each $x \in R_2$, we add the set $\{z_2(x)\} \times V$ to $E(H)$.
\end{description}

See Figure~\ref{fig:5cliqueConstruction} for an example of the whole construction. We now show that $H$ is a $5$-spanner of size  $O(n^{4/3} \log^{2/3} n)$.

\begin{lemma}
Let $u,v \in V$. There is a temporal path from $u$ to $v$ of length at most $5$ in $H$.
\end{lemma}
\begin{proof}
Let $x \in R_1$ be the center of the level-one cluster $C_{1,x}$ containing $u$ and $y \in R_2$ be the center of the level-two cluster $C_{2,y}$ containing $z_1(x)$. 

We first show that in $H$ there exists a temporal path $\pi$ of length $4$ from $u$ to $z_2(y)$ consisting of the sequence of edges $(u,x)$, $(x,z_1(x))$, $(z_1(x),y)$, $(y,z_2(y))$. Notice that, the edges $(u,x)$, $(x,z_1(x))$, $(z_1(x),y)$, $(y,z_2(y))$ belong to $ E_{1,u}$, $E_{1,z_1(x)}$, $E_{2,z_1(x)}$, and $E_{2,z_2(y)}$, respectively. Moreover, the initialization phase ensures that they all belong to $E(H)$. Then, by definition of $z_1(x)$, we have $\lambda(u,x) \leq \lambda(x,z_1(x))$. Moreover, since $(x,z_1(x)) \in E_{1,z_1(x)}$ and $(z_1(x),y) \in E_{2,z_1(x)}$, then $\lambda(x,z_1(x)) \leq \lambda(z_1(x),y)$. Finally, $(y, z_2(y)) \in E_{2,z_2(y)}$ and, by definition of $z_2(y)$, we have $\lambda(z_1(x),y) \leq \lambda(y,z_2(y))$.

If $v = z_2(x)$, then $u$ can reach $v$ via a temporal path of length $4$ in $H$, by using $\pi$. Moreover, if $v=z_1(x)$ then $u$ can reach $v$ via a temporal path of length $2$ by using the subpath $\pi_1$ of $\pi$ consisting of the edges  $(u,x)$ and $(x,z_1(x))$.
Otherwise we are in one of the following three cases:
\begin{itemize}
    \item If $(z_1(x),v) \in E_{1,z_1(x)}$, then $v \in S_{1,z_1(x)}$ and, due to the first augmentation phase, we have that $(u,v) \in E(H)$.
    \item If $(z_2(y),v) \in E_{1,z_2(y)} \cup E_{2,z_1(x)}$, then vertex $v \in S_{1,z_2(y)} \cup S_{2,z_1(x)}$ and the first augmentation phase ensures that $(z_1(x), v) \in E(H)$. Moreover, since$(z_1(x),v) \not\in E_{1,z_1(x)}$, we have that $\lambda(x,z_1(x)) \leq \lambda(z_1(x),v)$. Hence the concatenation of $\pi_1$ with the edge $(z_1(x),v)$ yields a temporal path of length $3$ from $u$ to $v$ in $H$.
    \item If $(z_2(y),v) \not\in (E_{1,z_2(y)} \cup E_{2,z_1(x)})$, the second augmentation phase ensures that $(z_2(y),v) \in E(H)$. Moreover, $\lambda(y,z_2(y)) \leq \lambda(z_2(y),v)$. Therefore the concatenation of $\pi$ with the edge $(z_2(y),v)$ yields a temporal path of length $5$ from $u$ to $v$ in $H$. \qedhere
\end{itemize}
\end{proof}

\begin{figure}
    \centering
    \includegraphics[scale=.9]{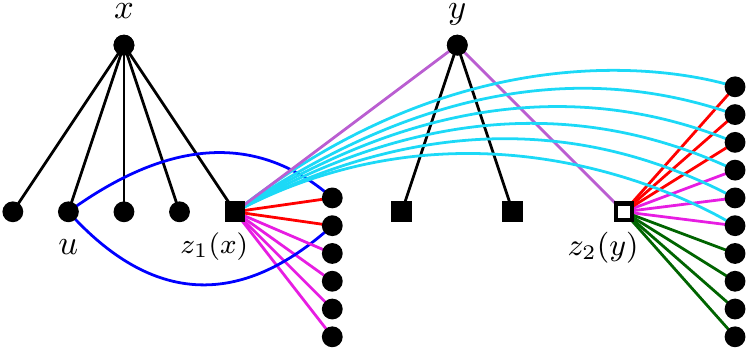}
    \caption{An example of a two-level cluster and of the edges added to $E(H)$ during the spanner construction. The black, red and purple edges are added during initialization phase. In particular, for every $u \in V$, the red edges are those in $E_{1,u}$ and, for every $z \in Z_1$, the purple edges are those in $E_{2,z}$. The dark blue and light blue edges are those added to $E(H)$ during the first augmentation phase. The green edges are the edges added to $E(H)$ during the second augmentation phase.}
    \label{fig:5cliqueConstruction}
\end{figure}

\begin{lemma}
\label{lemma:size_5spanner}
The size of $H$ is $O(n^{4/3} \log^{2/3} n)$.
\end{lemma}
\begin{proof}
During the initialization phase, for each $u \in V$, we add the $O(n^{1/3}\log^{2/3}n)$ edges in $E_{1,v}$  to $E(H)$.
Moreover, for each $u \in Z_1$, we add the $O(n^{2/3}\log^{1/3}n)$ edges in $E_{2,v}$ to $E(H)$ . Since $|Z_1| = O( n^{2/3} \log^{1/3} n )$, the total number of edges added during this phase is $O(n^{4/3} \log^{2/3} n)$.

During the first augmentation phase, for each $u \in V$  we add the edges in $E_{1,z_1(x)}$ to $E(H)$, where $x$ is the center of the cluster $C_{1,x}$ containing $u$.
Moreover, for each $u \in Z_1$ we add the edges in $\{u\} \times (S_{1,z_2(x)} \cup S_{2,z_2(x)})$ to $E(H)$, where $x$ is the center of the cluster $C_{2,x}$ containing $u$. Since $|E_{1,z_1(x)}| = O(n^{1/3}\log^{2/3}n)$, and $|S_{1,z_2(x)} | \le |Z_1| = | S_{2,z_2(x)}| = O(n^{2/3}\log^{1/3}n)$, the total number of edges added during this phase is $O(n^{4/3} \log^{2/3} n)$.

Finally, during the second augmentation phase, for each $ x \in R_2$, we add the edges in $\{z_2(x) \} \times V$ to $E(H)$. Since $|R_2| = O(n^{1/3}\log^{2/3}n)$, the total number of edges added during this phase is $O(n^{4/3} \log^{2/3} n)$.

By summing up the number of edges added during all phases, we obtain that the size of $H$ is $O(n^{4/3} \log^{2/3} n)$.
\end{proof}

\subsection{Our temporal \texorpdfstring{\boldmath $(2k-1)$}{(2k-1)}-spanner}
\label{sec:k-spanner}
\begin{algorithm}[t] \small
   	\caption{\small Computes a temporal $(2k-1)$-spanner.}
\label{alg:hier}
  	
    \SetKwInOut{Input}{Input}
    \SetKwInOut{Output}{Output}

    \Input{A temporal clique $G$;}
    \Output{A temporal $(2k-1)$-spanner of $G$;}
    
  	\BlankLine
    $Z_0 \gets V$\;
    \lForEach{$u \in V$}{$E(u) \gets \{(u,v) \mid v \in V\}$}
    \For{$i = 1,\dots,k-1$}
    {
        \ForEach{$u \in Z_{i-1}$}
        {
            $E_{i,u} \gets$ set of the first min-time label $n^{\frac{i}{k}}\log n$ edges of $E(u)$\;
            $E(u) \gets E(u) \setminus E_{i,u}$\;
            $S_{i,u} \gets \{v \in V : (u,v) \in E_{i,u} \}$\;
        }
        \BlankLine
        $R_i \gets $hitting set of $\{S_{i,u}\}_{u \in Z_{i-1}}$ computed as in Lemma~\ref{lemma:hitting_set}\;
        $C \gets \emptyset$  \tcp*{Set of vertices in $Z_{i-1}$ that are already clustered}
        \ForEach{$x \in R_i$}
        {
            $C_{i,x} = \{u \in Z_i \setminus C : x \in S_{i,u}\}$\;
            $z_i(x) \gets \arg\max_{u \in C_{i,x}}\{\lambda(u,x) \}$\;
            $C \gets C \cup C_{i,x}$\;
        }
        \BlankLine
        $Z_i \gets \{ z_i(x) \in Z_{i-1}: x \in R_i\}$\;
    }

    \BlankLine
    $H \gets (V,\emptyset)$
    \For(\tcp*[f]{Initialization}){$i = 1$ to $k-1$}{
        \lForEach{$u \in Z_{i-1}$}{
            $E(H) \gets E(H) \cup E_{i,u}$%
        }
    }
    
    \BlankLine
    \For(\tcp*[f]{First augmentation}){$i = 1, \dots, k-1$}
    {
        \ForEach{$u \in Z_{i-1}$}
        {
            Let $x \in R_i$ such that $u \in C_{i,x}$\;
            $E(H) \gets E(H)  \cup \{u\} \times \{ \bigcup_{j =1}^i S_{z_i(x),j}\}$\;
        }
    }
    \BlankLine
    
    \lForEach(\tcp*[f]{Second Augmentation}){$z \in Z_{k-1}$}{
        $E(H) \gets E(H)  \cup \{z\} \times V$%
    }
    \BlankLine
    \Return $H$\;
\end{algorithm}

In this section, we describe an algorithm that, given an integer $k \ge 2$ and a temporal clique $G$ of $n$ vertices, returns a temporal $(2k-1)$-spanner of $G$ with size $O(k \cdot n^{1+\frac{1}{k}} \log^{\frac{k-1}{k}}  n)$.

The idea is to define a hierarchical clustering of $G$, where a generic level-$i$ clustering partitions the special vertices of the level-$(i-1)$ clustering and identifies the special vertices of level-$i$.
As we move from one clustering level to the next, the number of clusters decreases by a factor of roughly $n^{\frac{1}{k}}$, thus allowing us to add an increasing number of edges incident to the special vertices into the spanner.

We ensure that each vertex  $u \in V$ can reach some special vertex by moving upwards in the clustering hierarchy. These special vertices work as hubs, i.e., each of them allows to directly reach a subset of vertices of $V$, and some special vertex of higher level (via a temporal path of length at most $2$). Then $u$ can reach any vertex in $v \in V$ by first reaching a suitable special vertex $z$ in the hierarchy, and then following the edge $(z,v)$.

We build our clustering in $k-1$ rounds indexed from $1$ to $k-1$ (a detailed pseudocode is given in Algorithm~\ref{alg:hier}), where the generic $i$-th round defines a set $Z_i$ of level-$i$ special vertices. Initially, $Z_0 = V$, i.e., all vertices are special vertices of level $0$.
During the $i$-th round, the level-$i$ clustering is computed from the set of vertices in $Z_{i-1}$ defined at the previous round as follows.
For each $u \in Z_{i-1}$, we let $E_{i,u}$ be a set of $\delta_i = \Theta(n^{\frac{i}{k}}\log^{\frac{k-i}{k}} n)$ edges with the smallest label among those that are incident to $u$ but do not belong to $\bigcup_{j=1}^{i-1} E_{u,j}$, and we denote by $S_{i,u} = \{ v \in V : (u,v) \in E_{i,u}\}$ the set containing the endvertices of the edges incident to $u$ in $E_{i,u}$. 
We now compute a hitting set  $R_i \subseteq V$ of the collection $\{S_{i,u} \mid u \in Z_{i-1} \}$ having size at most $O(\frac{n}{\delta_i} \log n)$. Lemma \ref{lemma:hitting_set} guarantees that $R_i$ always exists.
Notice that, as $i$ increases, the time labels of the edges in $E_{i,u}$ became larger, $\delta_i$ increases, and the $|R_i|$ decreases.

\begin{figure}[t]
    \centering
    \includegraphics[scale=.8]{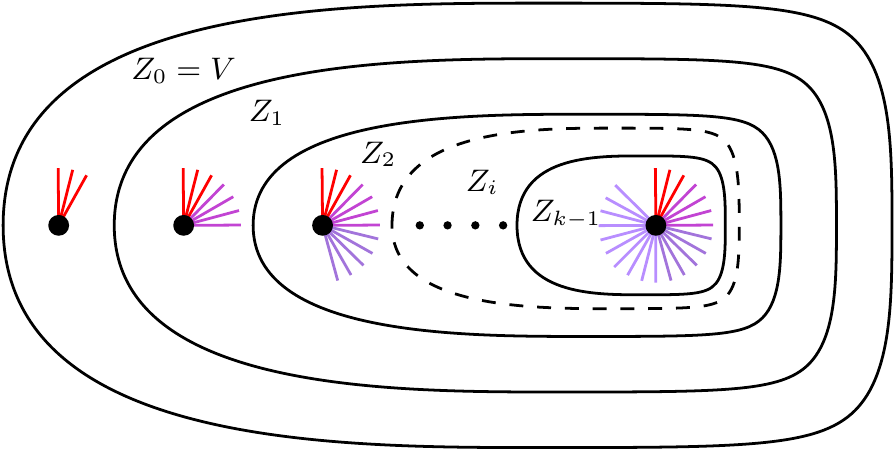}
    \caption{The set of edges selected for each vertex during the initialization phase.}
    \label{fig:hierarchy}
\end{figure}

We now partition the vertices in $Z_{i-1}$ into $|R_i|$ clusters $C_{i,x}$, one for each $x\in R_i$. We do so by adding each vertex $u \in Z_{i-1}$ into an arbitrary cluster $C_{i,x}$ such that $x \in S_{i,u}$. We call $x$ the center of the cluster $C_{i,x}$.
Moreover, for each cluster $C_{i,x}$, we choose a special vertex $z_i(x) \in C_{i,x}$ as a vertex that that maximizes the label of edge $(x,z_i(x))$.

Once the hierarchical clustering is built, our algorithm proceeds to construct a temporal $(2k-1)$-spanner $H$ of $G$. At the beginning $H = (V, \emptyset)$, then edges are added to $H$ in the following three phases:
\begin{description}
\item[Initialization:] For each $u \in V$, we add to $E(H)$ all the edges in the sets $E_{i,u}$ for $i = 1,\dots, j+1$, where $j$ is the largest integer between $0$ and $k-2$ for which $u \in Z_j$, see Figure~\ref{fig:hierarchy}. 
\item[First Augmentation:] For each $i = 1, \dots, k-1$ and each $u \in Z_{i-1}$, we consider the center $x \in R_i$ of the level-$i$ cluster $C_{i,x}$ containing $u$, and we add to $E(H)$ all the edges $(u,v)$ with $v \in \bigcup_{j =1}^i S_{j,z_i(x)}$.
\item[Second Augmentation:] We add to $E(H)$ all edges incident to some vertex in $Z_{k-1}$.
\end{description}

We now show that all vertices are at distance at most $2k-1$ in $H$, and that the size of $H$ is $O(k\cdot n^{1+\frac{1}{k}} \log^{\frac{k-1}{k}} n)$.

\begin{lemma}
\label{lemma:stretch_allpair}
For every $u,v \in V(G)$, $d_{H}(u,v) \leq (2k-1)d_{G}(u,v)$.
\end{lemma}
\begin{proof}
Let $z_0 = u$ and, for $i = 1,\ldots,k-1 $, let $z_i = z_i(x_{i})$ where $x_i \in R_i$ is the center of the cluster $C_{i,x_i}$ containing $z_{i-1}$. The initialization phase ensures that, for any $i$, there exists a temporal path from $z_0$ to $z_i$ in $H$ of length $2i$ entering $z_i$ with the edge $(x_i,z_i) \in E_{i,z_i}$. Indeed, $\pi_i$ can be chosen as the path that traverses edge $(z_{i-i},x_i) \in E_{i,z_{i-1}}$ and edge $(x_i,z_i) \in E_{i,z_i}$, in this order. Notice that, by definition of $z_i$, $\lambda(z_{i-1},x_i) \leq \lambda(x_i,z_{i})$. See Figure \ref{fig:path_pi}.

If  $v = z_i$ for some $i = 1,\ldots,k-1$ then, from the discussion above, we know that $\pi_i$ is a temporal path from $u$ to $v$ in $H$ of length $2i<2k-1$. Otherwise, we distinguish two cases depending on whether there exists some $i = 1, \ldots, k-1$ such that $(z_i,v) \in \bigcup_{j = 1}^{i}E_{j,z_i}$.

Suppose that the above condition is met, and let $i > 0$ be the minimum index for which $(z_i,v) \in \bigcup_{j = 1}^{i}E_{j,z_i}$.
If $\lambda(z_i,v) \geq \lambda(x_i,z_i)$, then $\pi_i$ followed by edge $(z_i,v)$, is a temporal path from $u$ to $v$ of length $2i+1\le 2k-1$.
If $\lambda(z_i,v) < \lambda(x_i,z_i)$ then, since $(z_i,v) \in \bigcup_{j = 1}^{i}E_{j,z_i}$, we have $v \in \bigcup_{j = 1}^{i}S_{j,z_i}$ and the first augmentation phase adds $(z_{i-1},v)$ to $E(H).$ By hypothesis we have $(z_{i-1},v) \not\in \bigcup_{j = 1}^{i-1}E_{j,z_{i-1}}$ and hence $\lambda(z_{i-1},v) \geq \lambda(x_{i-1},z_{i-1})$. This shows that $\pi_{i-1}$ followed by $(z_{i-1},v)$ is a temporal path from $u$ to $v$ in $H$ of length $2i-1\le 2k-1$.

It only remains to handle the case in which, for every $i$, we have $(z_i,v) \not\in \bigcup_{j = 1}^{i}E_{j,z_i}$.
In this case,  the algorithm adds $(z_{k-1},v)$ to $E(H)$
during the second augmentation phase. Moreover, since $\lambda(z_{k-1},v) \geq \lambda(x_{k-1},z_{k-1})$, the path $\pi_{k-1}$ followed by edge $(z_{k-1},v)$ is a temporal path from $u$ to $v$ in $H$ of length $2k-1$.
\end{proof}

\begin{figure}
    \centering
    \includegraphics[scale=0.80]{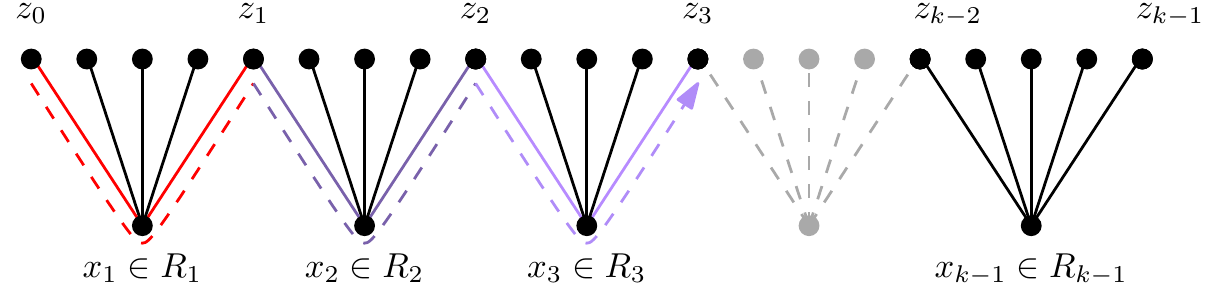}
    \caption{The hierarchy of clusters for vertex $z_0$, where $z_i = z_i(x_i)$ where $x_i \in R_i$ is such that $z_{i-1} \in C_{i,x_i}$. The dashed line is the temporal path $\pi_3$ that goes from $z_0$ to $z_3$.}
    \label{fig:path_pi}
\end{figure}

\begin{theorem}\label{thm:upper_bound_all_to_all}
Given a temporal clique $G$, for any $k \geq 1$, the above algorithm computes a temporal $(2k-1)$-spanner $H$ of size $O(k\cdot n^{1+\frac{1}{k}}\log^{\frac{k-1}{k}} n)$.
\end{theorem}
\begin{proof}
The upper bound on the stretch follows from Lemma~\ref{lemma:stretch_allpair}, therefore we focus on upper bounding the size of $H$.

Our algorithm constructs $H$ in three phases: initialization, first augmentation and second augmentation. In order to bound the size of $H$, we bound the number of edges added to $H$ in each phase.

For $i= 1$ to $k-1$, by Lemma~\ref{lemma:hitting_set}, we have that $|Z_i| = O\big( \frac{n}{\delta_i}\log n \big)$, since $\delta_i = \Theta( n^{\frac{i}{k}}\log^{\frac{k-i}{k}} n)$, then $|Z_i| = O( n^{1-\frac{i}{k}} \log^{i/k} n )$.

The number of edges added during the initialization phase is: 
\[
    \sum_{i = 1}^{k-1}|Z_{i-1}|\delta_i 
    = O\left( \sum_{i = 1}^{k-1} n^{1-\frac{i-1}{k}} \log^{\frac{i-1}{k}} n \cdot n^{\frac{i}{k}} \log^{\frac{k-i}{k}} n \right)
    = O\left( k \cdot n^{1+\frac{1}{k}}\log^{\frac{k-1}{k}} n \right).
\]

The number of edges added is during the first augmentation phase is:
\begin{align*}
    \sum_{i=1}^{k-1} &\left(|Z_{i-1}| \cdot \sum_{j=1}^i \delta_j \right)
    =  O\left( \sum_{i=1}^{k-1} |Z_{i-1}| \cdot \sum_{j=1}^i \Big(\frac{n}{\log n}\Big)^{\frac{j}{k}} \cdot \log n \right)  \\
    &= O\Bigg( \sum_{i =1}^{k-1} |Z_{i-1}|  \cdot  \frac{\left(\frac{n}{\log n}\right)^{\frac{i+1}{k}} -1}{\left(\frac{n}{\log n}\right)^{\frac{1}{k}} -1} \cdot \log n  \Bigg)
    = O\left( \sum_{i =1}^{k-1} |Z_{i-1}|  \cdot  \left(\frac{n}{\log n}\right)^{\frac{i}{k}}  \cdot \log n  \right) \\
     &= O\left(  \sum_{i =1}^{k-1} \left( n^{1+\frac{i-1}{k}} \log^{\frac{i-1}{k}} n \right) \cdot \left(\frac{n}{\log n}\right)^{\frac{i}{k}} \cdot \log n \right) 
    = O\left(  \sum_{i =1}^{k-1} n^{1+\frac{1}{k}}\log^{\frac{k-1}{k}} n \right) \\
    &= O\left(k \cdot n^{1+\frac{1}{k}}\log^{\frac{k-1}{k}} \right).
\end{align*}

Finally, the number of edges added during the second augmentation phase is:
\[
n \cdot |Z_{k-1}| 
= O\left( n \cdot n^{1-\frac{k-1}{k}}  \log^{\frac{k-1}{k}}n \right)
= O\left( n^{1+\frac{1}{k}} \log^{\frac{k-1}{k}}n \right).
\]

Thus the overall number of edges in $H$ is $O(k\cdot n^{1+\frac{1}{k}}\log^{\frac{k-1}{k}} n)$, as claimed.
\end{proof}

\subsection{Lower-bounds for temporal cliques}
\begin{figure}
    \centering
    \includegraphics[scale=0.80]{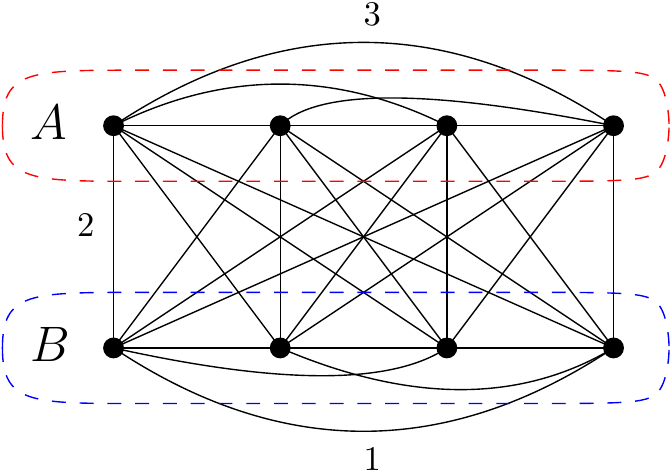}
    \caption{Temporal clique $G$, for which any temporal $2$-spanner has size $\Omega(n^2)$. All edges between vertices in $A$ have time-label $3$. All edges between vertices in $B$ have time-label $1$. All edges from vertices in $A$ to vertices in $B$ have time-label $2$.  }
    \label{fig:lb2}
\end{figure}
In this section, we give a lower bound on the size of temporal $2$-spanners of temporal cliques.

Consider a temporal clique $G$ whose vertices are partitioned into two sets $A$ and $B$ of size $n/2$. For any $(u,v) \in E(G)$, if $u,v \in A $, $\lambda(u,v)=3$, if $u,v \in B$, $\lambda(u,v)=1$ and if $u \in A$ and $v \in B$, $\lambda(u,v) = 2$, see Figure~\ref{fig:lb2}. By construction, all temporal paths that connect a vertex in $A$ to
a vertex in $B$ contain only edges with time-label $2$. Therefore, they all have odd length. Hence, if we remove any edge $(u,v)$ such that $u \in A$ and $v \in B$ then, the distance between $u$ and $v$ in the resulting graph is $3$. Since the number of edges between $A$ and $B$ is $\Omega(n^2)$, we have just shown the following:

\begin{theorem}
There exists a temporal clique $G$  of $n$ vertices such that any temporal $2$-spanner of $G$ has size $\Omega(n^2)$.
\end{theorem}

 Our construction can be slightly modified to show a lower bound of $\Omega(n^2)$ on the size of any temporal spanner for directed temporal cliques. Indeed, it suffices to consider each edge of the graph in Figure~\ref{fig:lb2} as bidirected, and replace the time-label of all the edges from $B$ to $A$ with $1$ (while the time-label of the edges from $A$ to $B$ remains $2$).
 
\section{Single-source spanners for general temporal graphs}
\label{sec:single_source}

In the first part of this section we design an algorithm that, for every $0 <\epsilon < n$, builds a single-source temporal  $(1+\epsilon)$-spanner of $G$ w.r.t.\ $s$ of size $O\big(\frac{n\log^4 n}{\log(1+\epsilon)}\big)$. We observe that, for constant values of $\epsilon$, the size of the computed spanner is almost linear, i.e., linear up to polylogarithmic factors.
The algorithm can be extended so as, for every $1 \leq \beta < n$, it builds a single-source temporal $\beta$-additive spanner of $G$ w.r.t. $s$ of size $O\big(\frac{n^2\log^4 n}{\beta}\big)$, see Section~\ref{section:additive} for details.

Our upper bounds leave open the problem of deciding whether a temporal graph $G$ admits a single-source temporal preserver w.r.t.\ $s$ of size $\softO(n)$. We answer to this question negatively in the second part of this section. More precisely, we show a temporal graph $G$ of size $\Theta(n^2)$ and a source vertex $s$ for which no edge can be removed if we want to keep a shortest temporal path from $s$ to every other vertex $u$. The construction can be extended to show a lower bound of $\Omega(n^2/\beta)$ on the size of single-source temporal $\beta$-additive spanners, for every $\beta \geq 1$. This implies that our upper bound on the size of single-source temporal additive spanners is asymptotically optimal, up to polylogarithmic~factors.

\subsection{Our upper bound}
\label{sec:single_source_ub}

In this section we present an algorithm that, for every $0<\epsilon<n$, computes a single-source temporal $(1+\epsilon)$-spanner of $G$ w.r.t. $s$ of size $O\big(\frac{n\log^4 n}{\log(1+\epsilon)}\big)$.\footnote{Our algorithm also works in the case of directed temporal graphs and/or multiple time-labels.}

In the following we say that a temporal path is {\em $\tau$-restricted} if it uses edges of time-label of at most $\tau$. Our algorithm computes a spanner that, for every $\tau = 1,\ldots,L$, contains $(1+\epsilon)$-approximate $\tau$-restricted temporal paths from $s$ to any vertex $v$ (recall that $L$ is the lifetime of $G$). More formally, for two vertices $u$ and $v$ of $G$, we denote by $\dle{G}{\tau}{u}{v}$ the length of a shortest $\tau$-restricted temporal path from $u$ to $v$ in $G$. We assume $\dle{G}{\tau}{u}{v}=+\infty$ when $G$ does not contain a $\tau$-restricted temporal path from $u$ to $v$. The single-source temporal  $(1+\epsilon)$-spanner $H$ of $G$ w.r.t.\ $s$ computed by our algorithm is such that, for every $v \in V$, and for every $\tau = 1,\ldots,L$, $\dle{H}{\tau}{s}{v} \leq (1+\epsilon)\dle{G}{\tau}{s}{v}$.

For technical convenience, in the following we design an algorithm that, for any $0 < \delta < n$ and any positive integer $k$, builds a single-source temporal $(1+\delta)^{k}$-spanner of $G$ w.r.t.\ $s$ of size $O\Big(\frac{kn^{1+\frac{1}{k}}\log^{2-\frac{1}{k}}n}{\log (1+\delta)}\Big)$. The desired bound of $O\big(\frac{n\log^4 n}{\log(1+\epsilon)}\big)$ on the size of the single-source temporal $(1+\epsilon)$-spanner is obtained by choosing $k=\lfloor \log n \rfloor$ and $\delta=(1+\epsilon)^{\frac{1}{\log n}}-1$. 

Our algorithm uses a subroutine that, for a given vertex $v$ of $G$, computes a set $\Pi_{v}$ of $O\Big(\frac{\log n}{\log(1+\delta)}\Big)$ temporal paths from $s$ to $v$ of $G$ such that, for every $\tau = 1,\ldots,L$, $\Pi_{v}$ contains a $\tau$-restricted temporal path $\pi$ satisfying $|\pi|\leq (1+\delta)\dle{G}{\tau}{s}{v}$. 
\setcounter{algocf}{1}
\begin{algorithm}[t] \small
   	\caption{\small Computes a set $\Pi_v$ of temporal paths from $s$ to $v$ in $G$ that provides a good approximation of any shortest $\tau$-restricted temporal path from $s$ to $v$ in $G$.}
\label{alg:single-pair}
  	
    \SetKwInOut{Input}{Input}
    \SetKwInOut{Output}{Output}

    \Input{A temporal graph $G$ (with lifetime $L$), a source vertex $s \in V$, and a vertex $v \in V$.}
    \Output{A set $\Pi_v$ of temporal paths from $s$ to $v$ in $G$ such that, for every $\tau = 1,\ldots,L$, there exists a $\tau$-restricted temporal path $\pi \in \Pi_v$ such that $|\pi|\leq (1+\delta)\dle{G}{\tau}{s}{v}$.}
    
  	\BlankLine

    {$\Pi_{v} \gets \emptyset$; $t \gets +\infty$}\;
    
    \For{$\tau = 1$ to $L$}
    {
        \If{$\dle{G}{\tau}{s}{v}\neq +\infty$ and $\dle{G}{\tau}{s}{v} < \frac{t}{1 +\delta}$}
        {
            Let $\pi$ be a shortest $\tau$-restricted temporal path from $s$ to $v$ in $G$\;
    
            $\Pi_v \gets \Pi_v \cup \{ \pi \}$\;
    
            $t \gets |\pi|$\;
        }
    }
    \Return $\Pi_v$\; 	
\end{algorithm}

The subroutine (see Algorithm~\ref{alg:single-pair} for the pseudocode) builds $\Pi_{v}$ iteratively by adding a subset of shortest $\tau$-restricted temporal paths from $s$ to $v$ in $G$, where $\tau = 1,\ldots,L$. We do so by scanning shortest $\tau$-restricted temporal paths from $s$ to $v$ in increasing order of values of $\tau$. The scanned path $\pi$ is added to $\Pi_{v}$ if no other path already contained in $\Pi_{v}$ has a length of at most $(1+\delta)|\pi|$. The next lemma shows the correctness of our subroutine and bounds the number of paths contained in $\Pi_{v}$. 

\begin{lemma}
\label{lemma:single-pair}
For every $\tau = 1,\ldots,L$, there is a $\tau$-restricted temporal path $\pi$ in $\Pi_{v}$ such that $|\pi|\leq (1+\delta)\dle{G}{\tau}{s}{v}$. Moreover, $|\Pi_{v}|=O\big(\frac{\log n}{\log (1+\delta)}\big)$.
\end{lemma}
\begin{proof}
Let $\pi^*$ be a shortest $\tau$-restricted temporal path from $s$ to $v$ in $G$ (we only need to consider the case in which $\pi^*$ exist). By construction, either $\Pi_{v}$ contains $\pi^*$ or there is $\tau' \le \tau$ such that $\Pi_{v}$ contains a shortest $\tau'$-restricted temporal path $\pi'$ from $s$ to $v$ in $G$ such that $|\pi'| \leq (1+\delta) |\pi^*|$. In either case, $\Pi_{v}$ contains a $\tau$-restricted path $\pi$ such that $|\pi|\leq (1+\delta)\dle{G}{\tau}{s}{v}$.

Let $\pi_1,\ldots \pi_h$ be the temporal paths in $\Pi_v$, in the order in which they are added to $\Pi_{v}$ by the algorithm and consider $h \ge 2$. By construction, for every $i=1,\dots, h-1$, we have $|\pi_{i}|>(1+\delta) |\pi_{i+1}|$, from which we derive $(1+\delta)^i|\pi_{i+1}| < |\pi_{1}|$. As $1 \leq |\pi_{i}|\leq n$ for every $i$, we have $(1+\delta)^{h-1}\leq (1+\delta)^{h-1}|\pi_{h}| < |\pi_{1}|\leq n$, from which we derive $h =O\Big(\frac{\log n}{\log (1+\delta)}\Big)$.
\end{proof}

In the rest of this section, for any given temporal path $\pi$, we denote by $\pi(\ell)$ the subpath of $\pi$ containing the last $\min\{\ell,|\pi|\}$ edges  of $\pi$. We observe that $\pi(\ell)=\pi$ when $|\pi|\leq \ell$. Moreover, for two vertices $u$ and $v$ of a temporal path $\pi$ that visits $u$ before $v$, we denote by $\pi[u,v]$ the temporal subpath of $\pi$ from $u$ to $v$.  

Before diving into the technical details, we describe the main idea of our algorithm and show how we can use it to build a single-source temporal $(1+\delta)^2$-spanner of $G$ w.r.t. $s$ of size $O\big(\frac{n^{3/2} \log^{3/2}n}{\log(1+\delta)}\big)$.

For technical convenience, let $R_0=V$ and $\mathcal{P}_0=\bigcup_{v \in R_0}\Pi_v$. In principle, we could build our single-source temporal $(1+\delta)$-spanner of $G$ w.r.t.\ $s$ by simply setting its edge set to $\bigcup_{\pi \in \mathcal{P}_0}E(\pi)$. Unfortunately, using only the result proved in Lemma~\ref{lemma:single-pair}, the  upper bound on the size of this spanner would already be quadratic in $n$. Therefore, to obtain a spanner of truly subquadratic size, we compute a single-source temporal $(1+\delta)^2$-spanner $H$ of $G$ w.r.t.\ $s$ instead. 

We build $H$ by adding all the {\em short} temporal paths in $\mathcal{P}_0$, i.e., all paths with at most $\ell_0$ edges for a suitable choice of $\ell_0$, and by replacing each {\em long} temporal path $\pi \in \mathcal{P}_0$ from $s$ to some vertex $v$ with the shortest temporal path from $s$ to $x$ in $\Pi_{x}$, for some vertex $x$ that {\em hits} $\pi(\ell_0)$, combined with $\pi[x,v]$. 

In more details, we define $\ell_0=\sqrt{n\log n}$ and we introduce a new parameter $\ell_1=n$. We say that a temporal path $\pi \in \mathcal{P}_0$ is {\em short} if $|\pi|\leq \ell_0$; it is {\em long} otherwise. Let $\mathcal{P}_0^{\textit{long}}=\{\pi \in \mathcal{P}_0 \mid |\pi|>\ell_0\}$ be the subset of long temporal paths in $\mathcal{P}_0$. We compute a set $R_1$ that hits $\{\pi(\ell_0) \mid \pi \in \mathcal{P}^{\textit{long}}_0\}$ using Lemma~\ref{lemma:hitting_set}, and we then use this set to define a new collection of temporal paths  $\mathcal{P}_1=\bigcup_{v \in R_1}\Pi_v$.\footnote{With a little abuse of notation, $R_1$ hits a temporal path $\pi$ if $R_1$ hits $V(\pi)$.} The edge set of $H$ is defined as $E(H)=\bigcup_{i\in\{0,1\}}\bigcup_{\pi \in \mathcal{P}_i}E(\pi(\ell_i))$. The next lemma shows that this simple algorithm already computes a single-source temporal $(1+\delta)^2$-spanner of $G$ w.r.t.\ $s$ of truly subquadratic size.

\begin{lemma}
\label{lemma:single-source_correctness_for_2}
For every $\tau = 1,\ldots,L$ and every $v \in V$, $\dle{H}{\tau}{s}{v} \leq (1+\delta)^2  \dle{G}{\tau}{s}{v}$. Moreover, the size of $H$ is $O\Big(\frac{n^{3/2}\log^{3/2} n}{\log(1+\delta)}\Big)$.
\end{lemma}
\begin{proof}
We start by proving the first part of the statement. Consider a vertex $v$ such that $\dle{G}{\tau}{s}{v}$ is finite and 
let $\pi \in \Pi_{v}\subseteq \mathcal{P}_0$ be the shortest $\tau$-restricted temporal path from $s$ to $v$ among those in $\mathcal{P}_0$. By Lemma~\ref{lemma:single-pair}, $|\pi| \leq (1+\delta)\dle{G}{\tau}{s}{v}$. If $\pi$ is short, then $\pi$ is entirely contained in $H$ and therefore $\dle{H}{\tau}{s}{v} = \dle{G}{\tau}{s}{v}\leq (1+\delta)^2 \dle{G}{\tau}{s}{v}$. Therefore, we henceforth assume that $\pi$ is long. Let $x \in R_1$ be a vertex that hits $\pi(\ell_0)$. By construction, the path $\pi[x,v]$, being a subpath of $\pi(\ell_0)$, is entirely contained in $H$. Let $\tau'$ be the time-label of the edge incident to $x$ in $\pi[x,v]$. Clearly, $\tau' \leq \tau$. Let $\pi' \in \Pi_{x}$ be a shortest $\tau'$-restricted temporal path from $s$ to $x$ among those in $\Pi_{x}$ (such a path always exists because $\pi[s,x]$ is $\tau'$-restricted). By construction, $\pi'$ is contained in $H$. Therefore, the concatenation of $\pi'$ with $\pi[x,v]$ is a $\tau$-restricted temporal path from $s$ to $v$ that is entirely contained in $H$. Moreover, using Lemma~\ref{lemma:single-pair}, we have $|\pi'| \leq (1+\delta)\dle{G}{\tau}{s}{x}\leq (1+\delta)|\pi[s,x]|$. As a consequence, $\dle{H}{\tau}{s}{v}\leq |\pi'|+|\pi[x,v]| \leq (1+\delta)|\pi[s,x]| + |\pi[x,v]| \leq (1+\delta) |\pi| \leq (1+\delta)^2\cdot\dle{G}{\tau}{s}{v}$.

To bound the size of $H$, we first observe that, for each $v \in V$, $|\Pi_{v}|=O\big(\frac{\log n}{\log(1+\delta)}\big)$ by Lemma~\ref{lemma:single-pair}. As a consequence, $\sum_{\pi \in \mathcal{P}_0}|\pi(\ell_0)|=\sum_{v \in R_0}\sum_{\pi \in \Pi_{v}}|\pi(\ell_0)|= O\big(n \ell_0 \frac{\log n}{\log(1+\delta)}\big)=O\big(\frac{n^{3/2}\log^{3/2} n}{\log(1+\delta)}\big)$. Moreover, by Lemma~\ref{lemma:hitting_set}, $|R_1|=O\big(\frac{n\log n}{\ell_0}\big)=O(\sqrt{n\log n})$. Therefore,  $\sum_{\pi \in \mathcal{P}_1}|\pi(\ell_1)|=\sum_{v \in R_1}\sum_{\pi \in \Pi_{v}}|\pi|=O\big(|R_1|n \frac{\log n}{\log(1+\delta)}\big)=O\big(\frac{n^{3/2}\log^{3/2} n}{\log(1+\delta)}\big)$. Hence, $|E(H)|=O\big(\frac{n^{3/2}\log^{3/2} n}{\log(1+\delta)}\big)$.
\end{proof}

\begin{algorithm}[t] \small
   	\caption{\small Single-source temporal spanner of a temporal graph $G$.}
\label{alg:single-source}
  	
    \SetKwInOut{Input}{Input}
    \SetKwInOut{Output}{Output}

    \Input{A temporal graph $G=(V,E)$ of $n$ vertices and a source vertex $s \in V$.}
    \Output{A single-source temporal spanner $H$ of $G$ w.r.t. $s$.}
    
  	\BlankLine
    \lFor{$i=0 \dots, k-1$}{$\ell_i \gets n^{\frac{i+1}{k}}\log^{1-\frac{i+1}{k}}n$}
    \lForEach{$v \in V$}{use Algorithm~\ref{alg:single-pair} to compute $\Pi_v$}

    \BlankLine

    $R_0 \gets V$; $\mathcal{P}_0 \gets \{\pi \in \Pi_v \mid v \in R_0\}$\;

    \BlankLine

    \For{$i = 1, \dots, k-1$}
    {
        $\mathcal{P}_{i-1}^{\textit{long}} = \{\pi \in \mathcal{P}_{i-1} \mid |\pi|>\ell_{i}\}$\;
        $R_{i} \gets$ hitting set of $\{\pi(\ell_{i-1}) \mid \pi \in \mathcal{P}_{i-1}^{\textit{long}}\}$ computed as in Lemma~\ref{lemma:hitting_set}\;
        $\mathcal{P}_{i} \gets \bigcup_{v \in R_{i}}\Pi_{v}$\;

    }

    \Return $H = \Big(V, \bigcup_{i=0}^{k-1} \bigcup_{\pi \in \mathcal{P}_i} E(\pi(\ell_i))\Big)$\; 	
\end{algorithm}

The technique we used to replace each of the temporal paths in $\mathcal{P}_0^{\text{long}}$ with a temporal path that is longer by a factor of at most $(1+\delta)$ can be applied recursively on the set $\mathcal{P}_1$, for a suitable choice of $\ell_1$, to obtain an even sparser spanner. As we show now,  $k-1$ levels of recursion allow us to compute a single-source temporal $(1+\delta)^{k}$-spanner $H$ of $G$ w.r.t.\ $s$ of size $O\Big(\frac{kn^{1+\frac{1}{k}}\log^{2-\frac{1}{k}}n}{\log(1+\delta)}\Big)$.

In the following we provide the technical details (see  Algorithm~\ref{alg:single-source} for the pseudocode). For every $i=0,\dots,k-1$, let $\ell_i=n^{\frac{i+1}{k}}\log^{1-\frac{i+1}{k}}n$. As before, let $R_0=V$ and $\mathcal{P}_0=\bigcup_{v \in R_0}\Pi_v$. During the $i$-th iteration, the algorithm computes a set $R_{i}$ that hits $\{\pi(\ell_{i-1}) \mid \pi \in \mathcal{P}_{i-1}^{\textit{long}}\}$, where $\mathcal{P}_{i-1}^{\textit{long}}=\{\pi \in \mathcal{P}_{i-1} \mid |\pi|>\ell_{i-1}\}$ is the set of {\em long} temporal paths of $\mathcal{P}_{i-1}$. The $i$-th iteration ends by computing the set $\mathcal{P}_{i}=\bigcup_{v \in R_{i}}\Pi_v$ that is used in the next iteration. 
The edge set of the graph $H$ that is returned by the algorithm is $E(H):=\bigcup_{i=0}^{k-1} \bigcup_{\pi \in \mathcal{P}_i} E(\pi(\ell_i))$.

\begin{theorem}\label{thm:single-source-multiplicative}
For every $\tau = 1,\ldots,L$, for every $i=1,\ldots,k$, and for every $v \in R_{k-i}$, we have that $\dle{H}{\tau}{s}{v} \leq (1+\delta)^{i} \dle{G}{\tau}{s}{v}$. Moreover, the size of $H$ is $O\Big(\frac{kn^{1+\frac{1}{k}}\log^{2-\frac{1}{k}}n}{\log (1+\delta)}\Big)$.
\end{theorem}
\begin{proof}
We start proving the first part of the theorem statement. The proof is by induction on $i$. Fix a vertex $v \in R_{k-i}$ such that $\dle{G}{\tau}{s}{v}$ is finite. 

For the base case $i=1$, we observe that $\Pi_v$ is entirely contained in $H$ by construction. Therefore, by Lemma~\ref{lemma:single-pair}, $\dle{H}{\tau}{s}{v} \leq (1+\delta) \dle{G}{\tau}{s}{v}$ and the claim follows.

We now prove the inductive case. We assume that the claim holds for $i-1$ and we prove it for $i$. Let $\pi \in \Pi_v$ be a shortest $\tau$-restricted temporal path from $s$ to $v$ among those in $\Pi_v$. By Lemma~\ref{lemma:single-pair}, $|\pi| \leq (1+\delta)\dle{G}{\tau}{s}{v}$. Moreover, by definition, $\pi \in \mathcal{P}_{k-i}$. If $\pi$ is short, i.e., $|\pi|\leq \ell_{k-i}$, then $\pi$ is entirely contained in $H$ and therefore $\dle{H}{\tau}{s}{v} \leq (1+\delta) \dle{G}{\tau}{s}{v} \leq (1+\delta)^{k-i}  \dle{G}{\tau}{s}{v}$. So, in the following we assume that $\pi$ is long. Let $x \in R_{k-i+1}$ be a vertex that hits $\pi(\ell_{k-i})$. By construction, the path $\pi[x,v]$, being a subpath of $\pi(\ell_{k-i})$, is entirely contained in $H$. Let $\tau'$ be the label of the edge incident to $x$ in $\pi[x,v]$. Clearly, $\tau' \leq \tau$. Moreover, by inductive hypothesis, $\dle{H}{\tau'}{s}{x} \leq (1+\delta)^{i-1}\cdot |\pi[s,x]|$. Let $\pi'$ be a shortest $\tau'$-restricted temporal path from $s$ to $x$ among those in $\Pi_x$ (such a path always exists because $\pi[s,x]$ is $\tau'$-restricted), and notice that $\pi'$ is contained in $H$. Therefore, the concatenation of $\pi'$ with $\pi[x,v]$ is a $\tau$-restricted temporal path from $s$ to $v$ that is entirely contained in $H$. As a consequence, $\dle{H}{\tau}{s}{v}\leq |\pi'|+|\pi[x,v]| \leq (1+\delta)^{i-1} \cdot |\pi[s,x]| + |\pi[x,v]| \leq (1+\delta)^{i-1} \cdot |\pi| \leq (1+\delta)^{i} \dle{G}{\tau}{s}{v}$.

To bound the size of $H$, we first observe that, for each $v \in V$,  $|\Pi_{v}|=O\big(\frac{\log n}{\log(1+\delta)}\big)$ by Lemma~\ref{lemma:single-pair}. Next, using Lemma~\ref{lemma:hitting_set}, we observe that each $R_i$, with $i\geq 1$, has size $|R_{i}|=O\big(\frac{n\log n}{\ell_{i-1}}\big)=O\big(n^{1-\frac{i}{k}}\log^{\frac{i}{k}} n\big)$. Furthermore, also $|R_0|=n=n^{1-\frac{0}{k}}\log^{\frac{0}{k}}n$. Therefore, for every $i=0,\ldots,k-1$, we have $|R_i|\ell_i=O\big(n^{1+\frac{1}{k}}\log^{1-\frac{1}{k}}n\big)$. As a consequence,
$$\sum_{\pi \in \mathcal{P}_i}|\pi(\ell_i)| =\sum_{v \in R_i}\sum_{\pi \in \Pi_v}|\pi(\ell_i)|=O\left(|R_i|\ell_i \frac{\log n}{\log(1+\delta)}\right) = O\left(\frac{n^{1+\frac{1}{k}} \log^{2-\frac{1}{k}}}{\log(1+\delta)} \right).
$$
Hence, $|E(H)|=\sum_{i=0}^{k-1}\sum_{\pi \in \mathcal{P}_i}|\pi(\ell_i)|=O\Big(\frac{kn^{1+\frac{1}{k}} \log^{2+\frac{1}{k}}}{\log(1+\delta)} \Big). 
$
\end{proof}

\noindent The following corollary follows by choosing $\tau = L$ and $i=k$ (so that $R_{k-i} = R_0 = V$):
\begin{corollary}
Let $G$ be a temporal graph with $n$ vertices and let $s$ be a vertex of $G$. The graph $H$ returned by Algorithm~\ref{alg:single-source} is a single-source temporal $(1+\delta)^k$-spanner of $G$ w.r.t. $s$ of size $O\Big(\frac{kn^{1+\frac{1}{k}} \log^{2+\frac{1}{k}}}{\log(1+\delta)} \Big)$.
\end{corollary}

\subsection{Extension to temporal additive-spanners}
\label{section:additive}
Our construction of Section~\ref{sec:single_source_ub} can be easily adapted to compute single-source temporal $\beta$-additive spanners of $G$ w.r.t. $s$ of size $O\big(\frac{n^2\log^4 n}{\beta}\big)$, for every $1 \leq \beta < n$. For technical convenience, we show how to adapt Algorithm~\ref{alg:single-source} so as, for any (not necessarily integral) value of $0 < \delta < n$ and any positive integer $k$, it builds a single-source temporal  $k\delta$-additive spanner of $G$ w.r.t. $s$ of size $O\Big(\frac{kn^{2+\frac{1}{k}}\log^{2-\frac{1}{k}}n}{\delta}\Big)$. The desired bound on the size of the single source temporal $\beta$-additive spanner is obtained by choosing $k=\lfloor\log n\rfloor$ and $\delta=\beta/\lfloor \log n \rfloor$.

We only need to extend Algorithm~\ref{alg:single-pair} so as it computes a set $\Pi_v$ of $O(n/\delta)$ temporal paths from $s$ to $v$ of $G$ such that, for every $\tau = 1,\ldots,L$, there exists a temporal path $\pi \in \Pi_v$ satisfying $|\pi| \leq \dle{G}{\tau}{s}{v}+\delta$. This can be done by modifying the if-statement so as the scanned path $\pi$ is added to $\Pi_v$ when $\dle{G}{\tau}{s}{v}\neq +\infty$ and $\dle{G}{\tau}{s}{v} < t+\delta$. No modification is required in the pseudocode of Algorithm~\ref{alg:single-source}. A proof similar to the one of Theorem~\ref{thm:single-source-multiplicative} allows us to obtain the following result.

\begin{theorem}
\label{thm:additive_single_source_spanner}
For every $\tau = 1,\ldots,L$, for every $i=1,\ldots,k$, and for every $v \in R_{k-i}$, we have that $\dle{H}{\tau}{s}{v} \leq \dle{G}{\tau}{s}{v} + i\delta$. Moreover, the size of $H$ is $O\Big(\frac{kn^{2+\frac{1}{k}}\log^{2-\frac{1}{k}}n}{\delta}\Big)$.
\end{theorem}

By choosing $\tau = L$ and $i=k$ we obtain the following corollary.

\begin{corollary}
Let $G$ be a temporal graph with $n$ vertices and let $s$ be a vertex of $G$. The graph $H$ returned by Algorithm~\ref{alg:single-source} is a single-source temporal  $k\delta$-additive spanner of $G$ w.r.t. $s$ of size $O\Big(\frac{kn^{2+\frac{1}{k}}\log^{2-\frac{1}{k}}n}{\delta}\Big)$. 
\end{corollary}

\subsection{Our lower bound}\label{sec:single-source_lb}

\begin{figure}
    \centering
    \includegraphics[width=\textwidth]{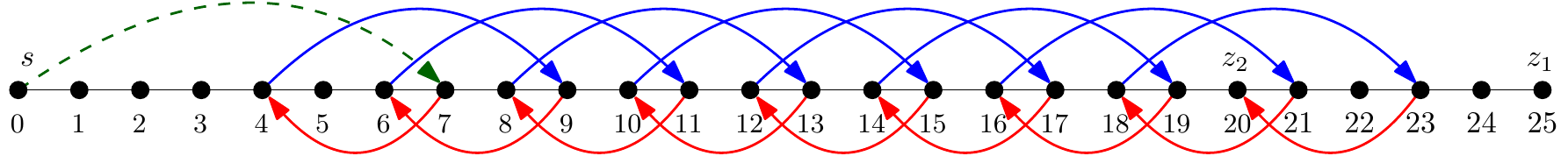}
    \caption{Example of the lower bound with $h=2$ and $\beta = 0$. Path $\pi_1$ consists of the black edges. Path $\pi_2$ is built on top of $\pi_1$, whose vertices have been numbered in order of a traversal from $s$, and consists in: the fist hop from $s$ to $\mu$ (in green), the forward hops (in blue), and the backward hops (in red). The arrows on the edges of $\pi_2$ are directed away from $s$ along $\pi_2$.}
    \label{fig:my_label}
\end{figure}
In this section we show that, for every $\beta \geq 0$, there is a temporal graph $G$ of $n$ vertices for which the size of any single-source temporal $\beta$-additive spanner of $G$ w.r.t.\ $s$ is $\Omega\big(\frac{n^2}{1+\beta}\big)$. This gives a lower bound of $\Omega(n^2)$ for the size of a single-source temporal preserver.

The temporal graph $G$ has $n=(13+\beta)h$ vertices, where $h$ is an integer, and is formed by the union of $h$ pairwise edge-disjoint temporal paths $\pi_1,\ldots,\pi_h$. Each path $\pi_i$ goes from $s$ to a vertex $z_i$ and has length $\Omega(n-i(1+\beta))$. The construction guarantees that the unique temporal path of $G$ from $s$ to $z_i$ of length of at most $d_G(s,z_i)+\beta$ is $\pi_i$.
This implies that the size of $G$ is $\Omega\big(\frac{n^2}{1+\beta}\big)$, as desired.

The temporal path $\pi_1$ is a Hamiltonian path that spans all the $n$ vertices of $G$ and goes from $s$ to $z_1$. All edges of $\pi_1$ have time-label $1$. The remaining temporal paths are defined recursively. More precisely, for each $i=2,\ldots,h$, the temporal path $\pi_i$ is defined on top of the temporal path $\pi_{i-1}$ as follows. Let us number the vertices visited in a traversal of $\pi_{i-1}$ from $s$ to $z_{i-1}$ in order from $0$ to $|\pi_{i-1}|-1$. The temporal path $\pi_i$ is defined as a sequence of hops over the vertices of $\pi_{i-1}$. We call {\em offset} a value $\mu$ that is equal to $\beta + 7$ for even values of $\beta$, and to $\beta + 8$ for odd values of $\beta$. The first hop is the one from $s$ to vertex $\mu$, if it exists. The rest of the path is given by a maximal alternating sequence of {\em backward} and {\em forward} hops that do not visit $z_{i-1}$. 
A generic backward hop goes from vertex $j$, with $j$ odd, to vertex $j-3$, while a generic forward hop goes from vertex $j$, with $j$ even, to vertex $j+5$. All the edges of $\pi_i$ have time-label $i$. A pictorial example of the definition of $\pi_i$ is given in Figure~\ref{fig:my_label}. The choice of odd values for the offset is a necessary condition to have pairwise edge-disjoint paths, while the dependency of the offset on $\beta$ guarantees that $\pi_i$ is the unique temporal path from $s$ to $z_i$ in $G$ such that $|\pi_i|\leq d_G(s,z_i)+\beta$. Finally, the alternating sequence of backward and forward hops guarantees that $|\pi_i|=\Omega(n-i(1+\beta))$.  The above discussion yields the following theorem, and a corollary for the case $\beta=0$.

\begin{theorem}
\label{thm:sslb}
For every positive integer $n$ and every $\beta\geq 0$, there is a temporal graph $G$ of $n$ vertices and a source vertex $s$ of $G$ such that any single-source temporal $\beta$-additive spanner of $G$ w.r.t.\ $s$ has size $\Omega\big(\frac{n^2}{1+\beta}\big)$.
\end{theorem}

\begin{corollary}
For every positive integer $n$, there is a temporal graph $G$ of $n$ vertices such that any single-source temporal preserver of $G$ w.r.t.\ $s$ has size $\Theta(n^2)$.
\end{corollary}

The remaining part of this section is devoted to proving Theorem~\ref{thm:sslb}. We start with some technical lemmas.

\begin{lemma}\label{lemma:single-source-lb-basic}
Vertex $j$ of $\pi_{i-1}$, with $j \geq \mu$, is at a distance of $1+j-\mu$ from $s$ in $\pi_i$ for all odd values of $j<|\pi_{i-1}|-1$ and it is at a distance of $5+j-\mu$ from $s$ in $\pi_i$ for all even values of $j\leq|\pi_{i-1}|-4$.
\end{lemma}
\begin{proof}
First of all, we observe that $\mu$ is odd. The claim for odd values of $j$ comes from the fact that if we are at a vertex $j$ and perform a backward hop followed by a forward hop, we end up at vertex $j-3+5=j+2$. Therefore, starting from $\mu$, we arrive at vertex $j \geq \mu$, with $j$ odd, after an alternating sequence of $j-\mu$ hops. Since $\mu$ is at a distance 1 from $s$ in $\pi_i$, vertex $j$ is at a distance of $1+j-\mu$ from $s$ in $\pi_i$.

For even values of $j<|\pi_{i-1}|-4$, we have that $j+3<|\pi_{i-1}|-1$ is odd. By the first part of this proof, $j+3$ is at a distance of $1+j+3-\mu=4+j-\mu$ from $s$ in $\pi_i$. Since $j$ is visited right after $j+3$ via a backward hop from $j+3$ to $j$, it follows that the distance from $j$ to $s$ in $\pi_i$ is $5+j-\mu$.
\end{proof}

Thanks to Lemma~\ref{lemma:single-source-lb-basic}, we can prove the following useful lemmas.

\begin{lemma}
\label{lemma:ssPathsLB}
$\sum_{i=1}^h|\pi_i|=\Omega\big(\frac{n^2}{1+\beta}\big)$.
\end{lemma}
\begin{proof}
We show that, for every $i$ with $1 \leq i \leq h$, $|\pi_{i}| \geq n-(i-1)\beta-13i$. 
This suffices to prove the claim since, using $h=\frac{n}{13+\beta}$, we have:
\begin{align*}
    \sum_{i=1}^h|\pi_i| &= nh - \frac{(h-1)h\beta}{2} - \frac{13h(h+1)}{2}
= nh - \frac{h^2 (\beta + 13)}{2} + \frac{h(\beta - 13)}{2}\\
&\ge nh - \frac{h^2 (\beta + 13)}{2} - \frac{13h}{2}
=  \frac{n^2}{2(\beta+13)} - \frac{13 n}{2(\beta+13)} 
= \frac{n^2 - 13 n}{2(\beta+13)}
= \Omega\Big(\frac{n^2}{1+\beta}\Big).
\end{align*}

The proof is by induction on $i$. The base case $i=1$ is trivial since $|\pi_1|=n-1$. We now assume that the claim holds for $\pi_{i-1}$ and prove it for $\pi_i$. We observe that the claim is proved once we show that $|\pi_i|\geq |\pi_{i-1}|-\beta-13$. By Lemma~\ref{lemma:single-source-lb-basic}, $\pi_i$ contains $s$ as well as all vertices of $\pi_{i-1}$ that are numbered from $\mu$ up to $|\pi_{i-1}|-5$. As $\mu\leq \beta+8$, we obtain
$|\pi_i|\geq |\pi_{i-1}|-5-\mu \geq |\pi_{i-1}|-\beta-13$.
\end{proof}

\begin{lemma}
\label{lemma:ssPathsUB}
For every $2 \leq i \leq h$  and for every $v \in V(\pi_i)$, with $v\neq s$, $|\pi_{i-1}[s,v]|>|\pi_{i}[s,v]| +\beta$.
\end{lemma}
\begin{proof}
Let $v \in V(\pi_i)$ be the vertex numbered $j$ in $\pi_{i-1}$. For $j=\mu-3$, we have that $|\pi_{i-1}[s,v]|=\mu-3 > 2+\beta=|\pi_{i}[s,v]| +\beta$. For all the other values of $j\geq \mu$, using the inequality $\mu \ge \beta + 7$ and Lemma~\ref{lemma:single-source-lb-basic}, we have
$|\pi_{i-1}[s,v]|=j > 5+j-\mu+\beta\geq|\pi_{i}[s,v]| +\beta$.
\end{proof}

In order to prove that the temporal paths $\pi_1,\ldots,\pi_h$ are edge-disjoint, we assign {\em types} to edges of each path and show that different paths have edges of different types if we restrict to edges that are not incident in $s$.

For each $ 1\leq i \leq h$ and $ 0 \leq j \leq |\pi_i|$, let $v_{i,j}$ be the $j$-th vertex of $\pi_i$, where we count vertices from $s$ to $z_i$. We enumerate each vertex $v$ of $G$ with $l(v)$ as follows: $l(v_{1,0}) = l(s) = 0$, $l(v_{1,j}) = l(v_{1,j-1}) +5$ for odd values of $j$, and $l(v_{1,j}) = l(v_{1,j-1}) - 3$ for even values of $j$. For an ordered pair of vertices $u$ and $v$ we define the value $\sigma(u,v) = l(v) - l(u)$. The type of an edge $(u,v)$ corresponds to the value $|\sigma(u,v)|$.
A pictorial example of the definition of $\sigma(u,v)$ can be found in  Figure~\ref{fig:ssLBTyping}. The following technical lemma is the key for proving that the paths $\pi_1,\dots,\pi_h$ are pairwise edge-disjoint.

\begin{figure}
    \centering
    \includegraphics[width=\textwidth]{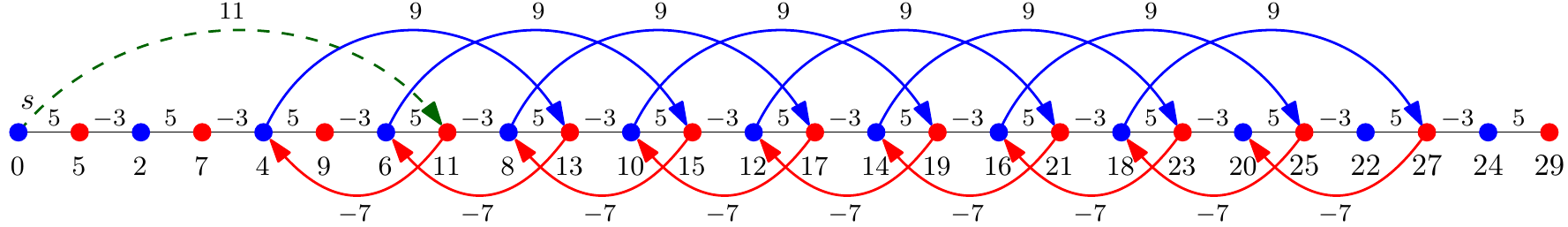}
    \caption{Example of the lower bound with $h=2$ and $\beta = 0$. Each vertex $v$ is labeled with $l(v)$. Path $\pi_1$ starts at $s$ and consists in all the black edges. The generic $j$-th vertex of $\pi_1$ is shown in blue if $j$ is even and in red if $j$ is odd. Path $\pi_2$ begins with an initial hop (the dashed green edge) from $s$ to the vertex $v$ with label $11$ (i.e.m the $7$-th vertex in $\pi_1$). Then $\pi_2$ consists of an alternating sequence of backward hops (shown in red) and forward hops (shown in blue).
    The arrows on the edges of $\pi_2$ are directed away from $s$ along $\pi_2$. For each traversed edge $(u,v)$ of $\pi_1$, $\sigma(u,v)$ is either $5$ or $-3$, while for each traversed edge $(u,v)$ of $\pi_2$, $\sigma(u,v)$ is either $-7$ or $9$.}
    \label{fig:ssLBTyping}
\end{figure}
\begin{lemma}
\label{lemma:disjoint}
For every $ 1\leq i \leq n$, and $2 \leq j < |V(\pi_i)|$, we have that $\sigma(v_{i,j-1},v_{i,j}) = -(4i-1)$ if $j$ is even and $\sigma(v_{i,j-1},v_{i,j}) = 4i+1$ if $j$ is odd.
\end{lemma}
\begin{proof}
The proof is by induction on $i$. The proof for the base case $i = 1$ directly follows from the definition (see also Figure~\ref{fig:ssLBTyping}).

For the inductive case, we assume that the claim holds for $i-1$ and prove it for $i$. Let $j$ be any integer such that $2\leq j < |V(\pi_i)|$ and consider the edge $(v_{i,j-1},v_{i,j})$ of $\pi_i$. This edge is  either a backward hop (even values of $j$) or a forward hop (odd values of $j$) over $\pi_{i-1}$. We now split the proof into two cases, according to value of $j$.

We first consider the case in which $j$ is even. The edge $(v_{i,j-1},v_{i,j})$ is a backward hop over $\pi_{i-1}$. So, for some $j'\geq 2$, $v_{i,j-1}=v_{i-1,j'}$ and $v_{i,j}=v_{i-1,j'-3}$.
As the value of the offset is always odd, backward hops over $\pi_{i-1}$ are always from vertices that are numbered with odd values in $\pi_{i-1}$. This implies that $j'$ is an odd value. Therefore, by inductive hypothesis, we have that:
\begin{itemize}
    \item $\sigma(v_{i-1,j'-1},v_{i-1,j'}) = 4(i-1)+1$;
    \item $\sigma(v_{i-1,j'-2},v_{i-1,j'-1}) = -(4(i-1)-1)$;
    \item $\sigma(v_{i-1,j'-3},v_{i-1,j'-2}) = 4(i-1)+1$.
\end{itemize}
As a consequence, by a repeated use of the definition of $\sigma$, we obtain $l(v_{i-1,j'}) = l(v_{i-1,j'-3}) +\sigma(v_{i-1,j'-3},v_{i-1,j'-2}) + \sigma(v_{i-1,j'-2},v_{i-1,j'-1}) + \sigma(v_{i-1,j'-1},v_{i-1,j'})= l(v_{i-1,j'-3})+ 4(i-1)+1 - (4(i-1)-1)+ 4(i-1)+1= l(v_{i-1,j'-3}) + 4i-1$. Hence, $
 \sigma(v_{i,j-1},v_{i,j}) = \sigma(v_{i-1,j'},v_{i-1,j'-3})= l(v_{i-1,j'-3})-l(v_{i-1,j'})=-(4i-1)$.
 
We move to the case in which $j$ is odd. The edge $(v_{i,j-1},v_{i,j})$ is a forward hop over $\pi_{i-1}$. So, for some $j'\geq 2$, $v_{i,j-1}=v_{i-1,j'}$ and $v_{i,j}=v_{i-1,j'+5}$. As the value of the offset is always odd, forward hops over $\pi_{i-1}$ are always from vertices that are numbered with even values in $\pi_{i-1}$. This implies that $j'$ is an even value. Therefore, by inductive hypothesis, we have that:
\begin{itemize}
    \item $\sigma(v_{i-1,j'},v_{i-1,j'+1}) = 4(i-1)+1$;
    \item $\sigma(v_{i-1,j'+1},v_{i-1,j'+2}) = -(4(i-1)-1)$;
    \item $\sigma(v_{i-1,j'+2},v_{i-1,j'+3}) = 4(i-1)+1$;
    \item $\sigma(v_{i-1,j'+3},v_{i-1,j'+4}) = -(4(i-1)-1)$;
    \item $\sigma(v_{i-1,j'+4},v_{i-1,j'+5}) = 4(i-1)+1$.
\end{itemize}
 By a repeated use of the definition of $\sigma$, we obtain $l(v_{i-1,j'+5}) = l(v_{i-1,j'}) +\sigma(v_{i-1,j'},v_{i-1,j'+1}) + \sigma(v_{i-1,j'+1},v_{i-1,j'+2}) + \sigma(v_{i-1,j'+2},v_{i-1,j'+3})+\sigma(v_{i-1,j'+3},v_{i-1,j'+4})+\sigma(v_{i-1,j'+4},v_{i-1,j'+5}) = l(v_{i-1,j'})+ 4(i-1)+1 - (4(i-1)-1)+ 4(i-1)+1- (4(i-1)-1)+ 4(i-1)+1= l(v_{i-1,j'}) + 4i+1$. Hence, $
 \sigma(v_{i,j-1},v_{i,j}) = \sigma(v_{i-1,j'},v_{i-1,j'+5})= l(v_{i-1,j'+5})-l(v_{i-1,j'})=4i+1$.
\end{proof}

\begin{lemma}\label{lm:edge_disjoint_paths}
The paths $\pi_1,\dots,\pi_h$ are pairwise edge-disjoint.
\end{lemma}
\begin{proof}
Consider any two distinct paths $\pi_i$ and $\pi_j$. W.l.o.g., we assume that $i < j$. By construction, each path $\pi_{k+1}$ spans a subset of vertices of $\pi_{k}$ but not $v_{k,1}$ (i.e., the vertex of $\pi_{k}$ adjacent to $s$). Therefore, the edge incident to $s$ of $\pi_i$ is different from the one incident to $s$ of $\pi_j$. For the remaining edges we use Lemma~\ref{lemma:disjoint}, which states that each edge of $\pi_i$ not incident to $s$ has a type in $\{4i-1,4i+1\}$, while each edge of $\pi_j$ that is not incident to $s$ has type in $\{4j-1,4j+1\}$. Since $4i+1 < 4(i+1)-1 \leq 4j-1$, we have $\{4i-1,4i+1\} \cap \{4j-1,4j+1\}=\emptyset$. Hence, $\pi_i$ and $\pi_j$ are edge-disjoint.
\end{proof}

The following lemma shows that any single-source temporal spanner of $G$ w.r.t.\ $s$ must contain all paths $\pi_1,\dots,\pi_h$. 

\begin{lemma}
\label{lemma:minimality}
For every $i\in\{1,\dots,h\}$, $\pi_i$ is the unique temporal path from $s$ to $z_i$ in $G$ such that $|\pi_i| \leq d_G(s,z_i)+\beta$. 
\end{lemma}
\begin{proof}
By construction, all edges that are incident to $z_i$ have a time-label less than or equal to $i$ (because $z_{i}$ is not spanned by $\pi_{i+1}$). Consider any temporal path $\pi$ from $s$ to $z_i$ in $G$ such that $\pi \neq \pi_i$. We prove the claim by showing that $|\pi|>d_G(s,z_i)+\beta$. Let $j$ be the time-label of the edge of $\pi$ that is incident to $s$ and let $v$ be the last vertex of $\pi$ that is incident to an edge of $\pi$ of time-label $j$. From $\pi \neq \pi_i$ it follows that $j < i$. Moreover, by construction of the paths $\pi_1,\dots,\pi_h$, $v$ is a vertex of $\pi_{j+1}$. As a consequence, the concatenation of the temporal path $\pi_{j+1}[s,v]$ with the temporal path $\pi[v,z_i]$ is a temporal path from $s$ to $z_i$. Using Lemma~\ref{lemma:ssPathsUB} we obtain $
    |\pi|=|\pi_j[s,v]|+|\pi[v,z_i]| > |\pi_{j+1}[s,v]|+\beta+|\pi[v,z_i]|\geq d_G(s,z_i)+\beta$.
\end{proof}

Since the paths $\pi_1,\dots,\pi_h$ are pairwise edge-disjoint by Lemma~\ref{lm:edge_disjoint_paths}, we have that Lemma~\ref{lemma:ssPathsUB} immediately implies $|E(G)|=\Theta\big(\frac{n^2}{1+\beta}\big)$. Furthermore, by Lemma~\ref{lemma:minimality}, we have that the unique single-source temporal $\beta$-additive spanner of $G$ w.r.t.\ $s$ is $G$ itself.
This proves~Theorem~\ref{thm:sslb}.
\section{Spanners for temporal graphs of bounded lifetime}
\label{sec:lifetime}
In this section we study how the lifetime has an impact on the size of both temporal single-source and temporal all-pairs spanners. We recall that the lifetime $L$ of $G$ is the number of distinct time-labels used for the edges of $G$ and that we assumed (w.l.o.g.) that the used time-labels are those in $\{1, \dots, L\}$. 

\subsection{Single-source preservers} 
For temporal single-source preservers we have asymptotically matching upper and lower bounds of $\Theta(Ln)$, for all values of $L \leq n$. The lower bound of $\Omega(Ln)$ is obtained by tweaking the construction given in Section~\ref{sec:single-source_lb} as follows. We set $\beta=0$ and the input temporal graph $G$ contains the paths $\pi_1,\dots,\pi_h$, with $h=\min\big\{L, \frac{n}{13}\big\}$. The upper bound of $O(Ln)$ follows from the fact that shortest $\tau$-restricted temporal paths satisfy the following suboptimality property: If $\pi$ is a shortest $\tau$-restricted temporal path from $s$ to $v$, $u\neq s$ is a vertex of $\pi$, and $\tau'$ is the time-label of the edge incident to $u$ in $\pi[s,u]$, then $\pi[s,u]$ is a shortest $\tau'$-restricted temporal path from $s$ to $u$. Indeed, thanks to this suboptimality property, for each vertex $v$ and each value $\tau = 1,\ldots,L$, it is enough to keep in $H$ a single edge $(u,v)$ of time-label $\tau$ for which there exists a shortest $\tau$-restricted temporal path from $s$ to $v$ in $G$ using $(u,v)$ (ties among edges with the same time-label can be broken arbitrarily). To summarize, we have the following result.

\begin{theorem}
For a given temporal graph $G$ of $n$ vertices and lifetime $L$ and a source vertex $s$ of $G$, we can compute a single-source preserver $H$ of $G$ w.r.t.\ $s$ of size $O(Ln)$. Moreover, for every $n$ and every $L \leq n$, there is a temporal graph $G$ of $n$ vertices and lifetime $L$ and a source vertex $s$ of $G$ such that any single-source temporal preserver of $G$ w.r.t.\ $s$ has size $\Omega(Ln)$.
\end{theorem}

\subsection{All-pairs spanners} 

We now turn our attention to all-pairs temporal spanners and present our results, some of which hold for any temporal clique only, while some others hold for any temporal graph. 

\subsubsection{General temporal graphs} 

In \cite{AxiotisF16}, the authors show a class of dense temporal graphs with $n$ vertices and lifetime $L=\Omega(n)$ for which any temporal spanner needs $\Omega(n^2)$ edges even if we only need to preserve the temporal reachability among all pairs of vertices. This leaves open the problem of understanding whether a temporal graph $G$ with lifetime $L=o(n)$ admits a temporal spanner of size $o(n^2)$. We answer this question affirmatively, by providing an upper bound of $O(Ln)$ on the size of $\floor{\log n}$-spanners. This result is obtained as a byproduct of the following more general result. 

\begin{theorem}
Let $\mathcal{A}$ be an algorithm that takes a (non temporal) graph $G'$ of $n'$ vertices as input and  computes an $(\alpha,\beta)$-spanner of $G'$ of size $O(f(n'))$. Then, for any temporal graph $G$ of $n$ vertices and lifetime $L$, we can use algorithm $\mathcal{A}$ to build a temporal $(\alpha,L\beta)$-spanner $H$ of $G$ having size $O(L\cdot f(n))$.
\end{theorem}
\begin{proof}
Let $G_i$ be the static graph such that $V(G_i) = V$ and $E(G_i) = \{e \in E \mid \lambda(e)  = i\}$.
For any $G_i$, let $H_i$ be the $(\alpha,\beta)$-spanner of $G_i$ computed by Algorithm $\mathcal{A}$. We now show that the temporal graph $H$ with $V(H)=V$ and $E(H) = \bigcup_{i \leq L} E(H_i)$ satisfies the conditions of the theorem statement. 

Clearly the size of $H$ is $O(L \cdot f(n))$. Let $u$ and $v$ be any two distinct vertices of $G$ and let $\pi$ be a shortest temporal path $u$ and $v$ in $G$. We decompose $\pi$ in $L$ subpaths $\pi_1,\ldots,\pi_L$, where, for each $i$, $\pi_i$ is the subpath of $\pi$ where all the edges in $\pi_i$ have time-label $i$. Notice that some $\pi_i$'s might be empty paths. By construction, each $\pi_i$ is entirely contained in $G_i$ and it is indeed a shortest path from, say, vertex $x$ to, say, vertex $y$. Therefore, $H_i$ contains a path $\pi_i'$ from $x$ to $y$ such that $|\pi_i'| \leq \alpha|\pi_i|+\beta$. Consider the path $\pi'$ obtained by concatenating all $\pi'_i$ for $i = 1, \ldots, L$ (some $\pi_i'$ might be empty paths).  Notice that $\pi'$ is entirely contained in $H$, Moreover, $|\pi'| = \sum_{i = 1}^L |\pi_i'| \leq \sum_{i = 1}^L (\alpha |\pi_i| +\beta) \leq \alpha |\pi| + L \cdot \beta$. Therefore, $H$ is a temporal $(\alpha,L\cdot\beta)$-spanner of $G$.
\end{proof}

Using the fact that static graphs $G$ admit $(2k-1)$-spanners of size $O(n^{1+1/k})$ for every positive integer $k$ (see \cite{AlthoferDDJS93}), we obtain the following corollary which, for $k=\floor{\log n}$, implies an upper bound of $O(Ln)$ on the size of a temporal $\floor{\log n}$-spanner of $G$.

\begin{corollary}
For a given temporal graph $G$ of $n$ vertices and lifetime $L$, we can build a temporal $\floor{\log n}$-spanner $H$ of $G$ of size $O(Ln)$.
\end{corollary}

\subsubsection{Temporal cliques}

In the following we show how to build a temporal $2$-spanner of size $O(n\log n)$ for temporal cliques $G$ of lifetime $L=2$. We observe that such result cannot be extended to larger values of $L$ due to the $\Theta(n^2)$ lower bound of  Lemma~\ref{alg:3-spanner} that holds for temporal cliques of lifetime $3$. We also show how to build  a temporal $3$-spanner of size $O(2^L n \log n)$ for temporal cliques with any lifetime $L$. For $L = o(\log n)$, this bound is better than the $O(n\sqrt{n \log n})$ bound given in Theorem~\ref{thm:upper_bound_all_to_all}.

\paragraph{Lifetime \texorpdfstring{\boldmath $2$}{2}}

We describe the algorithm that builds a temporal $2$-spanner $H$ of a temporal clique $G$ with lifetime $L=2$. The formal description of the algorithm is given in Algorithm~\ref{alg:2-spanner}.

We build $H$ incrementally, starting from a graph with no edges. The algorithm uses two sets $X$ and $Y$, both initially set to $V$, that model the set of pairs $\{(x,y) \mid x \in X, y \in Y, x \neq y\}$ that need to be {\em covered} by the algorithm, where a pair $(x,y)$ is covered once we add to $H$ a temporal path from $x$ to $y$ of length of at most 2. At each iteration, we add to $E(H)$ a set of $O(n)$ edges forming a star graph that reduces the number of uncovered pairs by a constant factor. This clearly implies that the while-loop is executed $O(\log n)$ times and, therefore, that the overall number of edges we add to $H$ is $O(n \log n)$. 

The star we select at each step is either centered at a vertex $x^* \in X$ or at a vertex $y^* \in Y$. In the former case we cover all pairs in $\{(x,y) \mid x \in X, y \in Y^*, x \neq y\}$, for a suitable choice of $Y^* \subseteq Y$ of size of at least $(|Y|-2)/2$ (this allows us to update $Y$ by removing $Y^*$ from it). Similarly, in the latter case we cover all pairs in $\{(x,y) \mid x \in X^*, y \in Y, x \neq y\}$, for a suitable choice of $X^* \subseteq X$ of size of at least $(|X|-2)/2$ (this allows us to update $X$ by removing $X^*$ from it).

At the end of the while loop there are at most $O(n)$ uncovered pairs $(x,y)$. These pairs are covered  via the addition of the corresponding edges to $H$. We now prove the correctness of the algorithm.

\begin{algorithm}[t] \small
   	\caption{\small Temporal $2$-spanner of a temporal clique with lifetime $2$.}
\label{alg:2-spanner}
  	
    \SetKwInOut{Input}{Input}
    \SetKwInOut{Output}{Output}

    \Input{A  temporal clique $G$ with lifetime $2$.}
    \Output{A temporal $2$-spanner $H$ of $G$ having size $O(n \log n)$.}
    
  	\BlankLine
    $X\gets V$; $Y\gets V$; $H \gets (V, \emptyset)$\;
    \While{$|X| > 2$ or $|Y| > 2$}
    {
        \If{there is $x^* \in X$ s.t. $Y^*=\{y \in Y\setminus \{x^*\} \mid \lambda(x^*,y)=2\}$ satisfies $|Y^*| \geq \frac{|Y|-2}{2}$}
        {
            $E(H)\gets E(H) \cup \{(x,y) \mid x \in X, y \in Y^*, x\neq y\}$\;
            $Y \gets Y \setminus Y^*$\;
        }\Else{
            Let $y^* \in Y$ s.t.             $X^*= \{x \in X\setminus\{y^*\} : \lambda(y^*,x) = 1\}$ satisfies $|X^*|\geq \frac{|X|-2}{2}$\;
            $E(H)\gets E(H) \cup \{(x,y) \mid x \in X^*, y \in Y, x\neq y\}$\;
            $X \gets X \setminus X^*$\;
        }
    }
    $E(H) \gets E(H) \cup \{(x,y) \mid x \in X, y \in Y, x \neq y\}$\;
    \Return $H$\;
\end{algorithm}

\begin{lemma}\label{lemma:lifetime_2}
At each iteration of the while-loop of Algorithm~\ref{alg:2-spanner}, at least one of the following two conditions holds:
\begin{enumerate}[(a)]
    \item there exists $x^* \in X$ such that the set $\{y \in Y \setminus \{x^*\} \mid \lambda(x^*,y)=2\}$ has size at least $(|Y|-2)/2 $;
    \item there exists $y^* \in Y$ such that the set $\{x \in X \setminus \{y^*\} \mid \lambda(x,y^*)=1\}$ has size at least $(|X|-2)/2$.
\end{enumerate}
\end{lemma}
\begin{proof}
We assume that (a) does not hold and prove that (b) must hold. Let $K$ be the number of edges with time-label $1$ of the form $(x,y)$, with $x \in X$ and $y \in Y$. If for every $x^* \in X$ the set $\{y \in Y \setminus \{x^*\} \mid \lambda(x^*,y)=2\}$ has size strictly smaller than $(|Y|-2)/2$, then the complementary set $\{y \in Y \setminus \{x^*\} \mid \lambda(x^*,y)=1\}$ has size at least $|Y|-(|Y|-2)/2-1\geq |Y|/2$ (the additional $-1$ comes from the fact that $x^*$ might be a vertex of $Y$). Therefore, $K\geq \frac{|X| \cdot |Y|}{2}$. Thus, on average, each vertex $y \in Y$ is connected to at least $K/|Y| > (|X|-2)/2$ vertices of $X$ with an edge of time-label 1. Hence, condition (b) must hold.
\end{proof}

\begin{theorem}
For a given a temporal clique $G$ of $n$ vertices and lifetime $L=2$, Algorithm~\ref{alg:2-spanner} computes a 2-spanner $H$ of $G$ of size $O(n\log n)$.
\end{theorem}
\begin{proof}
To prove the correctness of the algorithm, we need to show that every pair $(x,y)$, with $x,y \in V$ and $x \neq y$ is covered by the algorithm. Let $(x,y)$ be any such pair. We split the proof into two cases.

The first case occurs when there is an iteration of the while-loop such that either $x$ is removed from $X$ or $y$ is removed from $Y$. W.l.o.g., consider the first iteration in which this happens. Consider the case in which $y$ has been removed from $Y$. In this case the algorithm has added to $H$ a star centered at a vertex $x^*$ containing the edge $(x^*,y)$ of time-label 2 and the edge $(x,x^*)$ of time-label in $\{1,2\}$. The path of length 2 going from $x$ to $y$ via $x^*$ is indeed temporal. Consider now the case in which $x$ has been removed from $X$. In this case the algorithm has added to $H$ a star centered at a vertex $y^*$ containing the edge $(x,y^*)$ of time-label 1 and the edge $(y^*,y)$. Again, the path of length 2 going from $x$ to $y$ via $y^*$ is indeed temporal. In either case, $(x,y)$ is covered.

The second case occurs when neither $x$ is removed from $X$ nor $y$ is removed from $Y$. From Lemma~\ref{lemma:lifetime_2}, at each iteration of the while-loop, the algorithm either removes at least $(|X|-2)/2 \geq 1$ vertices from $X$ or it removes at least $(|Y|-2)/2 \geq 1$ vertices from $Y$. Therefore, the algorithm exits from the while-loop in a finite number of iteration and then adds to $H$ an edge between all pairs $(x,y)$ with $x\in X$ and $y \in Y$, among which there is the edge $(x,y)$ that covers the pair $(x,y)$.
 
We now bound the size of $H$. We observe that the algorithm adds $O(n)$ edges to $H$ during each iteration of the while-loop. Moreover, the number of iterations of the while-loop is $O(\log n)$ because, at each iteration, we have that either the size of $X$ decreases from $|X| \ge 3$ to at most $|X|-\frac{|X|-2}{2} \le \frac{|X|}{6}$, or the size of $Y$ decreases from $|Y| \ge 3$ to at most $\frac{|Y|}{6}$ (see Lemma~\ref{lemma:lifetime_2}).
Finally, the number of edges added to $H$ in the remaining instructions is $O(n)$, since at least one of $X$ and $Y$ has constant size. Therefore, the overall number of edges added to $H$ is $O(n \log n)$.
\end{proof}

\paragraph{Lifetime \texorpdfstring{\boldmath $L$}{L}}
\begin{algorithm}[t] \small
   	\caption{\small Temporal $3$-spanner for temporal cliques with lifetime $L$.}
\label{alg:3-spanner}
  	
    \SetKwInOut{Input}{Input}
    \SetKwInOut{Output}{Output}

    \Input{A  temporal clique $G$ with lifetime $L$.}
    \Output{A temporal $3$-spanner of $G$ having $O(2^L n \log n)$ size.}
    
  	\BlankLine
    $Y\gets V$; $H \gets (V, \emptyset)$\;
    \While{$|Y| > 2^L$}
    {
        Let $t^*$ be the largest integer for which there exists a vertex $x^* \in V$ such that the set $Y^* = \{y \in Y\setminus\{x^*\} \mid \lambda(x,y) = t^*\}$ satisfies $|Y^*| \geq (|Y|-1)/2^{t^*}$\;
        $E(H) \gets E(H) \cup \{ (x^*,y) \mid y \in Y^*\}$\label{ln:add_edges_xstar_y}\;
        \ForEach{$x \in V \setminus \{x^*\}$}{
            Pick $y \in Y^*$ such that $\lambda(x,y) \leq t^*$\tcp*{Vertex $y$ always exists (see Lemma~\ref{lemma:lifetime_L})}
            $E(H) \gets E(H) \cup \{ (x,y) \}$\;
        }
        $Y = Y \setminus Y^*$\;
    }
    $E(H) \gets E(H) \cup \{(x,y) \mid x \in V, y \in Y, x \neq y\}$\label{ln:add_all_missing_edges}\;
    \BlankLine
    \Return $H$\;
\end{algorithm}

We describe the algorithm that computes an $O(2^L n\log n)$-size $3$-spanner $H$ of a temporal clique $G$ with lifetime $L$. The pseudocode is given in Algorithm~\ref{alg:3-spanner}.

The algorithm is similar to the algorithm for temporal graphs $G$ with lifetime $2$. We build $H$ incrementally, starting from a graph with no edges. The algorithm uses a set $Y$, initially set to $V$, that models the set of pairs $\{(x,y) \mid x \in V, y \in Y, x \neq y\}$ that need to be {\em covered} by the algorithm via temporal paths of length of at most 3. At each iteration, we add to $E(H)$ a set of $O(n)$ edges and decrease the number of uncovered pairs from $|Y|$ to roughly $(1-1/2^L) |Y|$ in the worst case.
This implies that the overall number of iterations of the algorithm is $O(2^L \log n)$ and gives the desired bound on the size of the resulting $3$-spanner.

At each iteration the algorithm selects $O(n)$ temporal paths of length $3$ from all vertices in $V$ to a subset $Y^*  \subseteq Y$ such that $|Y^*|\geq (|Y|-1)/2^{L}$. The set $Y^*$ is chosen in this way: let $t^*$ be the largest integer such that there exists a vertex $x^* \in V$ for which the set $Y^*=\{y \in Y\setminus\{x\} \mid \lambda(x,y)=t^*\}$ has size at least $(|Y|-1)/2^t$. The algorithm adds to $H$ the edges $\{(x^*,y) \mid y \in Y^*\}$ and, for each $x \in V\setminus \{x^*\}$, it adds to $H$ an edge $(x,y)$ such that $y \in Y^*$ and $\lambda(x,y) \leq t^*$ (such a vertex $y$ always exists, as we show in the following).

The pairs that are not covered during the while-loop are  at most $O(2^Ln)$, and they are covered via the addition of the corresponding edges to $H$.

\begin{lemma}\label{lemma:lifetime_L}
At each iteration of the while-loop of Algorithm~\ref{alg:3-spanner}, for every  $x \in V$, there exists $t \leq L$  such that the set $\{y \in Y\setminus\{x\} \mid \lambda(x,y)=t\}$ has size at least $(|Y|-1)/2^{t}$. Moreover, if $t^*$ is largest integer for which there exists $x^* \in V$ such that the set $Y^*=\{y \in Y\setminus\{x^*\} \mid \lambda(x^*,y)=t^*\}$ has size at least $(|Y|-1)/2^{t^*}$, then, for each $x \in V\setminus \{x^*\}$, there exists a vertex $y \in Y^*$ such that $\lambda(x,y) \leq t^*$.
\end{lemma}
\begin{proof}
We start by proving the first part of the statement. Fix a vertex $x \in V$ and, for any $i$, let $Y_i=\{y \in Y\setminus\{x\} \mid \lambda(x,y)=i\}$. Clearly, $\sum_{i=1}^L |Y_i| \geq |Y|-1$. However, if $|Y_i|<(|Y|-1)/2^i$  for every $i$, then $\sum_{i=1}^L|Y_i| < (1-1/2^L)(|Y|-1)<|Y|-1$. Therefore, there must be at least one value of $t$ for which $|Y_t|\geq (|Y|-1)/2^t$.

We now prove the second part of the statement. 
If $t^*=L$ then any vertex $y \in Y^*$ satisfies $\lambda(x,y)\leq t^*$. Therefore, we assume that $t^*<L$. Let $x \in V\setminus\{x^*\}$ be fixed and let $Y_i=\{y \in Y\setminus\{x\} \mid \lambda(x,y)=i\}$. By the choice of $t^*$ we have that $|Y_i|<(|Y|-1)/{2^i}$ for every $i>t^*$. Therefore, $\sum_{i>t^*}|Y_i| < (|Y|-1)/2^{t^*}$, i.e., the edges incident to $x$ that have a time-label greater than $t^*$ are less than $(|Y|-1)/2^{t^*}$.
Since $|Y^*|\geq (|Y|-1)/2^{t^*}$, there must exist some edge $(x,y)$ with $y \in Y^*$ and $\lambda(x,y)\leq t^*$.
\end{proof}

\begin{theorem}
For a given a temporal clique $G$ of $n$ vertices and lifetime $L$, Algorithm~\ref{alg:3-spanner} computes a 3-spanner $H$ of $G$ of size $O(2^Ln\log n)$.
\end{theorem}
\begin{proof}
To prove the correctness of the algorithm, we need to show that every pair $(x,y)$, with $x,y \in V$ and $x \neq y$ is covered by the algorithm. Let $(x,y)$ be any such pair.
If $y$ there is no iteration of the while loop for which $y \in Y^*$, then
edge $(x,y)$ is explicitly added to $E(H)$ at line~\ref{ln:add_all_missing_edges} of Algorithm~\ref{alg:3-spanner}.
Otherwise, we can focus the (unique) iteration of the while-loop in which $y \in Y^*$. If $x=x^*$, the pair $(x,y)$ is covered since the algorithm adds  edge $(x,y)$ to $E(H)$ in line~\ref{ln:add_edges_xstar_y}. If $x\neq x^*$, the pair $(x,y)$ is covered because, by Lemma~\ref{lemma:lifetime_L}, there is a vertex $y' \in Y^*$ such that $\lambda(x,y')\leq t^*$. Then, the algorithm adds to $H$ the edges $(x^*,y')$ and $(x^*,y)$, both of time-label $t^*$, and the edge $(x,y')$ having time-label at most $t^*$. Therefore, the path of length 3 that goes from $x$ to $y$ via $y'$ and $x^*$ is a temporal path. 

The bound of $O(2^Ln\log n)$ on the size of $H$ follows from the fact that the algorithm adds $O(n)$ edges to $H$ during each iteration, together with the fact that the size of $Y$ is reduced from $|Y| > 2^L \ge 2$ to at most $|Y| - (|Y|-1)/2^L < |Y| - |Y|/2^{L+1} = |Y|(1-\frac{1}{2^{L+1}})$ as ensured by Lemma~\ref{lemma:lifetime_L}.
Then, the size of $Y$ at the end of the $i$-th iteration is at most $n (1-\frac{1}{2^{L+1}})^i \le n e^{-i/2^{L+1}}$ and hence there can be only $O(2^L \log n)$ iterations before the size of $Y$ becomes at most $2^L$ and the while loop ends.

\end{proof}

\section{Other distance measures}

This section is devoted to showing strong lower bounds on the size of temporal $\alpha$-spanners of temporal cliques, for definitions of distance that differ from the one we focused on in the rest of the paper. These distances are natural and have already been considered in other papers (see, e.g., \cite{CalamaiCM21,WuCHKLX14}).

\label{section:distance_measures}
\paragraph{Earliest Arrival Time.} The \emph{arrival time} of a (non-empty) temporal path traversing edges $e_1, e_2, \dots, e_k$ is the time label $\lambda(e_k)$ of the last edge of the path. The \emph{earliest arrival time} (EAT) distance from a vertex $u$ to a vertex $v \neq u$ is the minimum arrival time among all temporal paths from $u$ to $v$ (if $u$ and $v$ are not temporally connected their distance is infinite).
We now show that there are temporal cliques $G$ with $n$ vertices such that, regardless of $\alpha$, any temporal $\alpha$-spanner of $G$ (w.r.t.\ the EAT distance) must contain $\Theta(n^2)$ edges.
In \cite[Theorem~1]{AxiotisF16} it is shown that, for any integer $\eta \ge 2$, there exists a temporal graph $G_\eta$ with $3\eta$ vertices, at least $\eta^2 / 2$ edges, and lifetime smaller than $4\eta$,  such that all pairs of vertices are temporally connected and any connectivity preserving temporal subgraph must contain all but at most $5\eta$ edges.\footnote{For the sake of simplicity, we restated the results in Theorem~1 of \cite{AxiotisF16} using slightly worse constants.}

We can transform $G_\eta$ into a temporal clique $G$ with $n=3\eta$ vertices by adding all missing edges
$e$ and setting $\lambda(e)=M$, for some large value $M$. Consider then any temporal $\frac{M}{4\eta}$-spanner $H$ of $G$. For any pair of vertices $u,v$, $H$ must contain a temporal path from $u$ to $v$ with arrival time smaller than $\frac{M}{4\eta} \cdot 4\eta = M$, thus implying that $\pi$ cannot use any edge in $E(G) \setminus E(G_\eta)$. As a consequence, the existence of $H$ implies the existence of a connectivity-preserving temporal subgraph $H_\eta$ of $G_\eta$ such that $|E(H_\eta)| \le |E(H)|$. Since $|E(H_\eta)|$ must be at least $\eta^2 / 2 - 5\eta$, we have $|E(H)| \ge \eta^2 / 2 - 5\eta = \Omega(n^2)$.

We remark that a slight modification of the above construction also provides a lower bound of $\Omega(n^2)$ on the size of any temporal $n$-spanner also when we require the set of all time-labels to form a continuous interval from $1$ to the lifetime of $G$ (a qualitative example of this modification is shown in Figure~\ref{fig:eat_ldt_lower_bounds}~(a)).
Indeed, we can choose $M=72\eta^2$ and augment $G$ by adding $12\eta+1$ additional vertices $z_1, \dots, z_{12\eta+1}$ and all the missing edges. Since there are
more than $72\eta^2$ edges between pairs of vertices $z_i$, we can assign each time label between $1$ and $M$ to one such edge. All remaining edges have time-label $M$. The number of vertices of this modified construction is $n=15\eta+1$ and we have $\frac{M}{4\eta} = 18 \eta > n$, showing that any temporal $n$-spanner of (the augmented version of) $G$ must contain at least $\eta^2 / 2 - 5\eta = \Omega(n^2)$ edges.

\paragraph{Latest Departure Time.} The \emph{departure time} of a (non-empty) temporal path traversing edges $e_1, e_2, \dots, e_k$ is $\lambda(e_1)$ and the \emph{latest departure time} (LDT) distance from a vertex $u$ to a vertex $v \neq u$ is the maximum departure time  among all temporal paths from $u$ to $v$ (if $u$ and $v$ are not temporally connected their distance is $-\infty$). 

\begin{figure}
    \centering
    \includegraphics{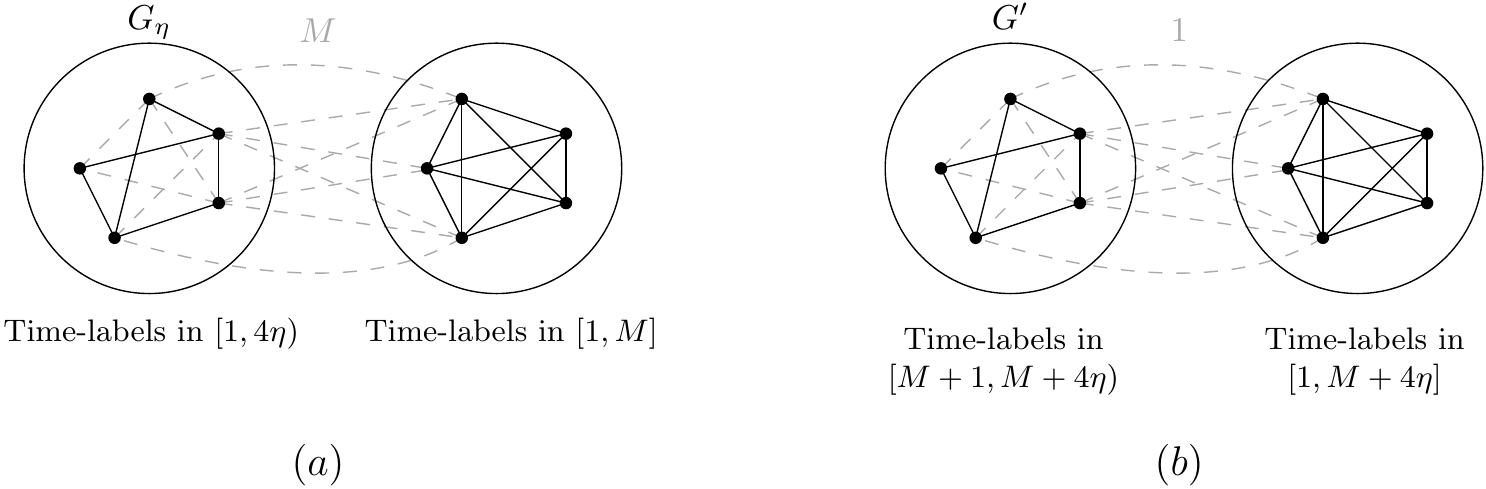}
    \caption{(a) A qualitative representation of our lower bound of $\Omega(n^2)$ on the size of any temporal $n$-spanner of a temporal clique w.r.t.\ the EAT distance. The solid black edges on the left side are exactly the graph $G_\eta$ constructed as in \cite[Theorem~1]{AxiotisF16}, and have time-labels in $[1, 4\eta)$. The clique on the right uses all time-labels in $[1, M]$ where $M=72 \eta^2$. All remaining edges have time-label $M$ (to avoid clutter, only some of these edges are shown using dashed gray lines). (b) A qualitative representation of our lower bound of $\Omega(n^2)$ on the size of any temporal $n$-spanner of a temporal clique w.r.t.\ the LDT distance. The construction is similar to the one used in (a), except that the time-labels in $G'$ are shifted by $M=\eta^2$, the clique on the right uses all time-labels in $[2, M+4\eta]$, and the remaining edges have time-label $1$.}
    \label{fig:eat_ldt_lower_bounds}
\end{figure}

It is known \cite{CalamaiCM21} that it is possible to reassign the time-labels of a temporal graph $G$ to obtain a new temporal graph $G'$ such that a temporal path from $u$ to $v$ with the latest departure time in $G$ corresponds to a temporal path from $v$ to $u$ with the earliest arrival time in $G$, and vice-versa. 
We can then use this transformation on the previous construction on the EAT distance to provide a lower bound of $\Omega(n^2)$ on the number of edges needed by temporal preservers of temporal cliques.

Considering approximate distances, we observe that (contrarily to the other distance measures in this section), a suboptimal path from $u$ to $v \neq u$ has a departure time that is \emph{smaller} than the LDT distance from $u$ to $v$, that is a path becomes more desirable as its departure time increases. We can then update our definition of temporal $\alpha$-spanner as follows: A subgraph $H$ of a temporal graph $G$ is a temporal $\alpha$-spanner of $G$ if, for every pair of distinct vertices $u,v$, it holds $d^{(LDT)}_H(u,v) \ge d^{(LDT)}_G(u,v) / \alpha$, where $d_H^{(LDT)}(u,v)$ (resp. $d_G^{(LDT)}(u,v)$) denotes the LDT distance from $u$ to $v$ in $H$ (resp. $G$).

We can once again use the lower bound construction of \cite{AxiotisF16} to show that there are temporal cliques $G$ on $n$ vertices for which all $\alpha$-spanner (regardless of $\alpha$) must have size $\Omega(n^2)$. Let $M$ be a large enough integer and consider the graph $G_\eta$ from \cite[Theorem~1]{AxiotisF16}, as in the discussion for the EAT distance. We consider the graph $G'$ obtained from $G_\eta$ by shifting all time-labels by $M$ so that they are between $M+1$ and $M+4\eta$ (clearly, this does not alter the lower bound of $\Omega(\eta^2)$ on the size of any connectivity-preserving temporal subgraph of $G$).
Our temporal clique $G$ is obtained from $G'$ by adding all missing edges $e$ and setting $\lambda(e)=1$. Then, for any pair of distinct vertices $u,v$, we have $d_G^{(LDT)}(u,v) > M$ showing that any temporal $M$-spanner $H$ of $G$ must contain a temporal path from $u$ to $v$ that does not use any edge with time-label $1$. As a consequence, the edges in $E(H) \cap E(G')$ induce a connectivity-preserving temporal subgraph of $G'$ and hence we must have $|E(H)| \ge  \eta^2/2 - 5 \eta = \Theta(n^2)$.

If we insist on using all time-labels between $1$ and the lifetime of $G$, we can employ a modification similar to the one used for the EAT distance. We are then able to show that there are temporal cliques $G$ on $n$ vertices for which $\Omega(n^2)$ edges are needed by any temporal $o(n^2)$-spanner of $G$. To do so we can set $M=\eta^2$ and augment the temporal graph $G$ constructed above with $\lceil \sqrt{2(\eta^2 + 4\eta)} \rceil +1 = \Theta(\eta)$ additional vertices. We add all edges between these new vertices and assign all time-labels between $1$ and $M+4\eta$ to them. We complete the graph using edges with time-label $1$. See Figure~\ref{fig:eat_ldt_lower_bounds}~(b) for a qualitative example.

\paragraph{Fastest Time.} The \emph{duration} of a (non-empty) temporal path traversing edges $e_1, e_2, \dots, e_k$ is the difference $\lambda(e_k) - \lambda(e_1)$ between the time-labels of the last and first temporal edges of the path. The \emph{fastest time} (FT) distance from a vertex $u$ to a vertex $v \neq u$ is the minimum duration among all temporal paths from $u$ to $v$ (if $u$ and $v$ are not temporally connected their distance is infinite). 
According to this definition, the duration of a path consisting of a single temporal edge is $0$. Hence, regardless of the value of $\alpha$, all temporal $\alpha$-spanner (w.r.t.\ the FT distance) must contain all $\Theta(n^2)$ edges of any temporal clique on $n$ vertices in which each temporal edge has a distinct time label.

Other models of temporal graph assign both a time-label $\lambda(e)$ and a non-negative travel time $\tau(e)$ to each temporal edge $e$. Here a path traversing edges $e_1, e_2, \dot, e_k$ is temporal if $\lambda(e_i) + \tau(e_i) \le \lambda(e_{i+1})$ for all $i=1, \dots, k-1$, and its duration is defined as $\lambda(e_k) + \tau(e_k) - \lambda(e_1)$. The above lower bound corresponds to the case in which $\tau(e)=0$ for all edges $e$. We observe that even if we restrict ourselves to temporal cliques in the case in which all travel times are positive, we still have a lower bound of $\Omega(n^2)$ on the size of any connectivity preserving subgraph, as it can be seen by considering a temporal clique in which all edges have time-label $1$.

\paragraph{Shortest Time.} 

If each temporal edge $e$ has an associated non-negative travel time $\tau(e)$, as discussed above, then it is possible to extend the concept of travel time to temporal paths. Specifically, the \emph{travel time} of a temporal path traversing edges $e_1, e_2, \dots, e_k$ is defined as $\sum_{i=1}^k \tau(e_i)$. The \emph{shortest time} (ST) distance from a vertex $u$ to a vertex $v \neq u$ is the minimum travel time among all temporal paths from $u$ to $v$ (if $u$ and $v$ are not temporally connected their distance is infinite). 
Also in this case, any connectivity preserving subgraph of a temporal clique in which all temporal edges have positive travel times and time-label $1$ must contain all $\Omega(n^2)$ edges.

\section{Conclusions}

In this paper we addressed the size-stretch trade-offs for temporal spanners. We showed that a temporal clique admits a temporal $(2k-1)$-spanner of size $\softO(n^{1+\frac{1}{k}})$, which implies a spanner having size $\softO(n)$ and stretch $O(\log n)$.
The previous best-known result was the temporal-spanner of \cite{CasteigtsPS21} which only preserves temporal connectivity between vertices. Our construction guarantees $O(\log n)$-approximate distances at the cost of only an additional $O(\log n)$ multiplicative factor on the size.
We also considered the single-source case for \emph{general} temporal graphs, where we provided almost-tight size-stretch trade-offs, along with the special case of temporal graphs with bounded lifetime.

The main problem that remains open is understanding whether better trade-offs are achievable for temporal cliques. In particular, no superlinear lower bounds are known even for the case of $3$-spanners.

Finally, as we already mentioned, temporal graphs admit other natural notions of distances between vertices  (which have have been used, e.g., in~\cite{MertziosMS19,WuCHKLX14,HuangFL15}). The most commonly used distances are the \emph{earliest arrival time}, the \emph{latest departure time}, the \emph{fastest time} (i.e., the smallest difference between the arrival and departure time of a temporal path from $u$ to $v$), and ---if each edge has an associated travel time--- the \emph{shortest time} distance (i.e., the minimum sum of the travel times of the edges of a temporal path from $u$ to $v$). 
One can wonder whether sparse temporal spanners with low stretch are attainable also in the case of the above distances. Section~\ref{section:distance_measures} provides a negative answer by showing strong lower bounds on the size of temporal $\alpha$-spanners for temporal cliques, even for large values of $\alpha$.

\bibliographystyle{plainurl}
\bibliography{bibliography}

\begin{thebibliography}{10}

\bibitem{AhmedBSHJKS20}
Abu~Reyan Ahmed, Greg Bodwin, Faryad~Darabi Sahneh, Keaton Hamm, Mohammad
  Javad~Latifi Jebelli, Stephen~G. Kobourov, and Richard Spence.
\newblock Graph spanners: {A} tutorial review.
\newblock {\em Comput. Sci. Rev.}, 37:100253, 2020.
\newblock \href {https://doi.org/10.1016/j.cosrev.2020.100253}
  {\path{doi:10.1016/j.cosrev.2020.100253}}.

\bibitem{AkridaGMS17}
Eleni~C. Akrida, Leszek Gasieniec, George~B. Mertzios, and Paul~G. Spirakis.
\newblock The complexity of optimal design of temporally connected graphs.
\newblock {\em Theory Comput. Syst.}, 61(3):907--944, 2017.
\newblock \href {https://doi.org/10.1007/s00224-017-9757-x}
  {\path{doi:10.1007/s00224-017-9757-x}}.

\bibitem{AlthoferDDJS93}
Ingo Alth{\"{o}}fer, Gautam Das, David~P. Dobkin, Deborah Joseph, and
  Jos{\'{e}} Soares.
\newblock On sparse spanners of weighted graphs.
\newblock {\em Discret. Comput. Geom.}, 9:81--100, 1993.
\newblock \href {https://doi.org/10.1007/BF02189308}
  {\path{doi:10.1007/BF02189308}}.

\bibitem{AxiotisF16}
Kyriakos Axiotis and Dimitris Fotakis.
\newblock On the size and the approximability of minimum temporally connected
  subgraphs.
\newblock In Ioannis Chatzigiannakis, Michael Mitzenmacher, Yuval Rabani, and
  Davide Sangiorgi, editors, {\em 43rd International Colloquium on Automata,
  Languages, and Programming, {ICALP} 2016, July 11-15, 2016, Rome, Italy},
  volume~55 of {\em LIPIcs}, pages 149:1--149:14. Schloss Dagstuhl -
  Leibniz-Zentrum f{\"{u}}r Informatik, 2016.
\newblock \href {https://doi.org/10.4230/LIPIcs.ICALP.2016.149}
  {\path{doi:10.4230/LIPIcs.ICALP.2016.149}}.

\bibitem{CalamaiCM21}
Marco Calamai, Pierluigi Crescenzi, and Andrea Marino.
\newblock On computing the diameter of (weighted) link streams.
\newblock In David Coudert and Emanuele Natale, editors, {\em 19th
  International Symposium on Experimental Algorithms, {SEA} 2021, June 7-9,
  2021, Nice, France}, volume 190 of {\em LIPIcs}, pages 11:1--11:21. Schloss
  Dagstuhl - Leibniz-Zentrum f{\"{u}}r Informatik, 2021.
\newblock \href {https://doi.org/10.4230/LIPIcs.SEA.2021.11}
  {\path{doi:10.4230/LIPIcs.SEA.2021.11}}.

\bibitem{CasteigtsPS21}
Arnaud Casteigts, Joseph~G. Peters, and Jason Schoeters.
\newblock Temporal cliques admit sparse spanners.
\newblock {\em J. Comput. Syst. Sci.}, 121:1--17, 2021.
\newblock \href {https://doi.org/10.1016/j.jcss.2021.04.004}
  {\path{doi:10.1016/j.jcss.2021.04.004}}.

\bibitem{CasteigtsRRZ20}
Arnaud Casteigts, Michael Raskin, Malte Renken, and Viktor Zamaraev.
\newblock Sharp thresholds in random simple temporal graphs.
\newblock In {\em Proceedings of the 62nd Annual {IEEE} Symposium on
  Foundations of Computer Science (FOCS)}. {IEEE}, 2022.
\newblock Full version at https://arxiv.org/abs/2011.03738.

\bibitem{Holme18}
Petter Holme.
\newblock Temporal networks.
\newblock In Reda Alhajj and Jon~G. Rokne, editors, {\em Encyclopedia of Social
  Network Analysis and Mining, 2nd Edition}. Springer, 2018.
\newblock \href {https://doi.org/10.1007/978-1-4939-7131-2\_42}
  {\path{doi:10.1007/978-1-4939-7131-2\_42}}.

\bibitem{HuangFL15}
Silu Huang, Ada~Wai{-}Chee Fu, and Ruifeng Liu.
\newblock Minimum spanning trees in temporal graphs.
\newblock In Timos~K. Sellis, Susan~B. Davidson, and Zachary~G. Ives, editors,
  {\em Proceedings of the 2015 {ACM} {SIGMOD} International Conference on
  Management of Data, Melbourne, Victoria, Australia, May 31 - June 4, 2015},
  pages 419--430. {ACM}, 2015.
\newblock \href {https://doi.org/10.1145/2723372.2723717}
  {\path{doi:10.1145/2723372.2723717}}.

\bibitem{KempeKK02}
David Kempe, Jon~M. Kleinberg, and Amit Kumar.
\newblock Connectivity and inference problems for temporal networks.
\newblock {\em J. Comput. Syst. Sci.}, 64(4):820--842, 2002.
\newblock \href {https://doi.org/10.1006/jcss.2002.1829}
  {\path{doi:10.1006/jcss.2002.1829}}.

\bibitem{MertziosMS19}
George~B. Mertzios, Othon Michail, and Paul~G. Spirakis.
\newblock Temporal network optimization subject to connectivity constraints.
\newblock {\em Algorithmica}, 81(4):1416--1449, 2019.
\newblock \href {https://doi.org/10.1007/s00453-018-0478-6}
  {\path{doi:10.1007/s00453-018-0478-6}}.

\bibitem{WuCHKLX14}
Huanhuan Wu, James Cheng, Silu Huang, Yiping Ke, Yi~Lu, and Yanyan Xu.
\newblock Path problems in temporal graphs.
\newblock {\em Proc. {VLDB} Endow.}, 7(9):721--732, 2014.
\newblock URL: \url{http://www.vldb.org/pvldb/vol7/p721-wu.pdf}, \href
  {https://doi.org/10.14778/2732939.2732945}
  {\path{doi:10.14778/2732939.2732945}}.

\end{thebibliography}

\end{document}